\documentclass[reqno]{amsart}
\usepackage{comment}
\usepackage{amssymb,latexsym,enumerate,upref}

\usepackage{finalT}

\usepackage{hyperref}

\numberwithin{equation}{section}

\DeclareMathOperator{\tr}{tr}
\newcommand{\id}{\mathord{{\mathrm 1}\kern-0.27em{\mathrm I}}\kern0.35em}

\DeclareMathOperator{\diag}{diag}

\DeclareMathOperator{\dsl}{\not{\!\partial}\!}
\newcommand{\del}[1]{\partial_{#1}}

\newcommand{\uset}[2]{\underset{#1}{#2}{}}

\newcommand{\AND}{{\quad\text{and}\quad}}

\newcommand{\Half}{\ensuremath{\textstyle\frac{1}{2}}}

\newcommand{\norm}[1]{\|#1\|}

\newcommand{\ipe}[2]{ ( #1 | #2 )}

\newcommand{\ip}[2]{ \langle #1 | #2 \rangle}

\newcommand{\Ac}{\mathcal{A}{}}
\newcommand{\Bc}{\mathcal{B}{}}
\newcommand{\Cc}{\mathcal{C}{}}

\newcommand{\Ec}{\mathcal{E}{}}
\newcommand{\Fc}{\mathcal{F}{}}

\newcommand{\Hc}{\mathcal{H}{}}

\newcommand{\Lc}{\mathcal{L}{}}

\newcommand{\Nc}{\mathcal{N}{}}
\newcommand{\Qc}{\mathcal{Q}{}}
\newcommand{\Rc}{\mathcal{R}{}}

\newcommand{\Tc}{\mathcal{T}{}}
\newcommand{\Uc}{\mathcal{U}{}}
\newcommand{\Vc}{\mathcal{V}{}}
\newcommand{\Wc}{\mathcal{W}{}}
\newcommand{\Xc}{\mathcal{X}{}}

\newcommand{\gt}{\tilde{g}{}}
\newcommand{\Gt}{\tilde{G}{}}

\newcommand{\Jt}{\tilde{J}{}}

\newcommand{\Mt}{\tilde{M}{}}

\newcommand{\Tt}{\tilde{T}{}}
\newcommand{\ut}{\tilde{u}{}}
\newcommand{\Ut}{\tilde{U}{}}

\newcommand{\Yt}{\tilde{Y}{}}

\newcommand{\sigmat}{\tilde{\sigma}{}}

\newcommand{\phit}{\tilde{\phi}{}}

\newcommand{\gammat}{\tilde{\gamma}{}}
\newcommand{\Gammat}{\tilde{\Gamma}{}}

\newcommand{\zetat}{\tilde{\zeta}{}}

\newcommand{\gb}{\bar{g}{}}
\newcommand{\Gb}{\bar{G}{}}

\newcommand{\nb}{\bar{n}{}}
\newcommand{\Rb}{\bar{R}{}}

\newcommand{\Tb}{\bar{T}{}}

\newcommand{\Yb}{\bar{Y}{}}

\newcommand{\Gammab}{\bar{\Gamma}{}}

\newcommand{\nablab}{\bar{\nabla}{}}

\newcommand{\zetab}{\bar{\zeta}{}}

\newcommand{\phih}{\hat{\phi}{}}

\newcommand{\lambdah}{\hat{\lambda}{}}

\newcommand{\Ebb}{\mathbb{E}{}}

\newcommand{\Mbb}[1]{{\mathbb{M}_{#1\times #1}}}
\newcommand{\Nbb}{\mathbb{N}{}}

\newcommand{\Rbb}{\mathbb{R}{}}
\newcommand{\Sbb}[1]{{\mathbb{S}_{#1}}}
\newcommand{\Tbb}{\mathbb{T}{}}
\newcommand{\Zbb}{\mathbb{Z}{}}

\newcommand{\Jch}{\check{J}{}}

\newcommand{\xv}{\mathbf{x}{}}
\newcommand{\Xv}{\mathbf{X}{}}

\newcommand{\phiv}{\boldsymbol{\phi}{}}

\newcommand{\phivt}{\tilde{\boldsymbol{\phi}}{}}

\newcommand{\ep}{\epsilon}

\newcommand{\Jtch}{\check{\tilde{J}}}

\newcommand{\Uct}{\widetilde{\mathcal{U}}{}}
\newcommand{\Vct}{\widetilde{\mathcal{V}}{}}
\newcommand{\Wct}{\widetilde{\mathcal{W}}{}}

\begin{document}

\title[Dynamical compact elastic bodies in general relativity]{Dynamical compact elastic bodies in general relativity}

\author[L. Andersson]{Lars Andersson}
\address{Albert Einstein Institute, Am M\"uhlenberg 1, D-14476 Potsdam,
  Germany}
\email{lars.andersson@aei.mpg.de}

\author[T.A. Oliynyk]{Todd A. Oliynyk}
\address{School of Mathematical Sciences\\
Monash University, VIC 3800\\
Australia}
\email{todd.oliynyk@sci.monash.edu.au}

\author[B. Schmidt]{Bernd G. Schmidt}
\address{Albert Einstein Institute, Am M\"uhlenberg 1, D-14476 Potsdam,
  Germany}
\email{bernd@aei.mpg.de}


\begin{abstract}
We prove the local existence of solutions to the Einstein-Elastic equations that represent
self-gravitating, relativistic elastic bodies with compact support.
\end{abstract}

\maketitle

\sect{intro}{introduction}

It is now almost 100 years since Einstein formulated the field equations of his theory of gravity, General Relativity. A little more than a month after the publication in November 1915, Schwarzschild found the spherically symmetric vacuum solution. It would however take more than half a century before its role as a black hole was recognized.

Since this first non-trivial existence result, much effort has been dedicated to establishing the existence of solutions to the Einstein field equations. As in any theory, the purpose of existence theorems are to demonstrate that the equations of the theory admit solutions capable of describing the physical situations which the theory is supposed to model. The purpose of this paper is to show that General Relativity has solutions describing dynamical, compact elastic bodies.

In 1952, Yvonne Choquet-Bruhat proved the first existence theorem for the vacuum Einstein field equations \cite{ChoquetBruhat:1952}, and since then, much  insight has been gained into the properties of the vacuum field equations, and also the field equations
coupled to field sources such as Maxwell or Yang-Mills. However, if we consider dynamical, isolated material bodies, the situation is quite different from the vacuum case. Almost nothing is known in General Relativity about general solutions describing isolated  material bodies, with the exception of some special cases involving either spherical symmetry, very specific equations of state, or the assumption of stationarity   \cite{Andersson_et_al:2008,Andersson_et_al:2010,ChoquetFriedrich:2006,KindEhlers:1993,Rendall:1992}.

This is in stark contrast to the experimental setting where gravitational interactions have been long studied
through observing the motion of isolated bodies, i.e. planets, moons, asteroids, comets and stars. Indeed, Kepler deduced the elliptical motion of the planets from Tycho Brahe's  observations, and Newton explained the elliptic orbits by his law of gravity. All of this was accomplished without having a theory of dynamical bodies. An important role here is played by the remarkable
property of Newton's theory of gravity, that for two-body systems the motion of the center of mass of the bodies is independent of the internal structure of the body.

Although it is possible in Newtonian gravity to calculate the gravitational field generated
by dynamical, isolated material bodies without any detailed information about the material bodies themselves, the
evolution of the material bodies must still be determined in order to have a complete
understanding of such systems. From a theoretical point of view, this requires having suitable local existence theorems for the systems of equations
that govern gravitating elastic bodies. Apart from the few special results \cite{Andersson_et_al:2008,Andersson_et_al:2010,ChoquetFriedrich:2006,KindEhlers:1993,Rendall:1992}
in General Relativity, the only local existence results that are available are in Newtonian gravity, and even in this simplified setting,
there remains much work to be done in order to cover all physically realistic situations.
Some key results in Newtonian gravity are as follows. In the approximation of a compact (non-fluid) elastic body moving in an external gravitational field, where the
gravitational self-interaction and interaction with the object
generating the external field are ignored,
local existence and uniqueness has been established
in \cite{ShibataNakamura:1989}.  Local existence and uniqueness results for the general case, which includes
gravitational self and mutual interactions between adjacent (non-fluid) elastic bodies, are given in \cite{Andersson_et_al:2011}.
Related results covering self-gravitating, isolated incompressible fluid bodies are given in \cite{LindbladNordgren:2009}. There
are also more local existence results available for isolated fluid bodies when gravitational
interactions are ignored, for example, see \cite{CoutandShkoller:2007,CoutandShkoller:2012,Coutand_et_al:2013,Lindblad:2005a,Lindblad:2005b}.

The difficulty, in both Newtonian gravity and General Relativity, in establishing local existence and uniqueness results
for elastic bodies can be attributed to
two sources,
the free boundary arising from the evolving matter-vacuum interface, and the irregularity in the stress-energy
tensor across the matter-vacuum interface. There are essentially two distinct types of irregularities.
The first type corresponds to gaseous fluid
bodies where the proper energy density monotonically decreases in a neighborhood of the vacuum boundary and vanishes identically
there. In this situation, the fluid evolution equations become degenerate and are no longer hyperbolic at the boundary leading to
severe analytic difficulties. In general, the energy density of the fluid is continuous but not differentiable at the boundary of the body in this situation. The second type of irregularity is where the proper energy density has
a positive limit at the vacuum boundary. Examples of this type are liquid fluids and solid elastic bodies. In this case, the energy density of the matter has a
jump discontinuity across the vacuum boundary. Such jump discontinuities
appear in the gravitational field equations in both Newtonian gravity and General Relativity. Dealing with
this discontinuity leads to a number of technical difficulties, which are particularly
acute in the General Relativistic setting due to the non-linearity of the Einstein field equations.

In this article, we restrict ourselves to considering (solid) elastic bodies in General Relativity having
a jump discontinuity in the matter across the vacuum boundary. In this situation, the boundary is characterized by the vanishing of the normal stress. Since the matter density is positive at the boundary, further conditions
on the gravitational field,
which we refer to as \textit{compatibility conditions}, must be imposed beyond the vanishing
of the normal stress in order to establish the existence of solutions\footnote{The elastic field must also satisfy compatibility conditions, but these are
of the well understood type that occur in boundary value problems and are often referred to as ``corner conditions''.} . These compatibility conditions, which have, as far as we know, never been observed up to now\footnote{See however  \cite{vanelst_et_al:2000PhRvD..62j4023V} for a  discussion of the general geometric conditions which the curvature tensor must satisfy for metrics where the second and third derivatives are discontinuous across a timelike surface.}, arise from matching conditions across the matter-vacuum boundary for the higher time derivatives of the gravitational field. The time derivatives $\partial^\ell_t g_{ij}$, or $\partial^\ell_t U$ in the Newtonian case, are determined by the Cauchy data at both sides of the boundary. Only if these time derivatives satisfy appropriate conditions at the boundary up to some order will that gravitational field be sufficiently differentiable inside and outside. In the physical spacetime dimension, the continuity of the time derivatives $\partial^\ell_t g_{ij}$ across the boundary is a consequence of the compatibility conditions.

The compatibility conditions impose restrictions on the Cauchy data. For static solutions they are trivially satisfied. Given
such data, our main aim is to
extend the local existence results of \cite{Andersson_et_al:2011} to the relativistic domain.
Under certain technical assumptions on the elasticity tensor, see Section \ref{matpic},
we establish the local existence of solutions to the Einstein-Elastic equations
that represent self-gravitating, relativistic elastic bodies with compact support. A precise formulation of this result is
given in Theorem \ref{fulllocthm}, which contains the main result of this article. As far as we are aware,
Theorem \ref{fulllocthm} represents the first local existence result that produces solutions without
any special assumptions, such as spherical symmetry,
that correspond to relativistic, gravitating compact matter sources having a jump discontinuity at the matter-vacuum interface.
Related local existence results for non-gravitating, relativistic elastic matter have been established in
\cite{BeigWernig:2007}; see also \cite{Oliynyk:CQG_2012,Trakhinin:2009} for the case of relativistic fluids.


\sect{EinElas}{Einstein-Elastic equations}

A single, compact, n-dimensional\footnote{The important case is $n=3$, which is the physical dimension. However,
since it is of no greater difficulty to handle the general case, we consider also all of higher dimensions
$n\geq 4$.}
$(n\geq 3)$,  relativistic elastic body is\footnote{The extension to non-colliding multiple interacting bodies
is straightforward.}, locally in time, characterized by a map
\eqn{fmap}{
f \: : \: W_T \longrightarrow \Omega
}
from a space-time cylinder $W_T \cong [0,T)\times \Omega$ to an open, bounded
subset $\Omega$ of $\Rbb^n$ with smooth boundary, known as the \emph{material space}, see \cite{BeigSchmidt:2003} for details. The body world tube
$W_T$ is taken  to be contained in an ambient Lorentzian spacetime $(M_T,g)$, where $M_T \cong [0,T)\times \Sigma$ for
some $n$-manifold $\Sigma$. For simplicity, we assume that $\Omega$ and $M_T$ are each covered
by a single coordinate chart given by $(X^I)$, $I=1,2,\ldots,n,$ and $(x^\lambda)$, $\lambda=0,1,\ldots,n$, respectively\footnote{We consistently
use lower case Greek letter, e.g. $\mu,\nu,\gamma$, to label spacetime coordinate
indices that run from $0$ to $n$ while capital Roman indices, e.g. $I,J,K$, which run
from $1$ to $n$, will be used to label the spatial material coordinates.}, and in the following, we use
\leqn{pardef}{
\del{I}=\frac{\del{}\;}{\del{}X^I} \AND \del{\lambda}=\frac{\del{}\;}{\del{}x^\lambda}
}
to denote partial differentiation with respect to these coordinates.

In these local coordinates, $f$ and $g$ can be decomposed as
\eqn{fg}{
X^I=f^I(x^\lambda) \AND g=g_{\mu\nu}(x^\lambda)dx^\mu dx^\nu,
}
respectively, and the field equations satisfied by $\{f^I,g_{\mu\nu}\}$, which we will refer to as
the \emph{Einstein-Elastic equations}, are given
by
\lalin{EinElasA}{
G^{\mu\nu} &= 2\kappa T^{\mu\nu} \text{\hspace{0.3cm} in $M_T$,} \label{EinElasA.1}  \\
\nabla_{\mu} T^{\mu\nu} &= 0 \text{\hspace{1.1cm} in $W_T$,} \label{EinElasA.2}
}
where $G^{\mu\nu}$ is the Einstein tensor of the metric $g_{\mu\nu}$, and
\eqn{Tdef}{
T_{\mu\nu} = 2\frac{\del{}\rho}{\del{} g^{\mu\nu}} - \rho g_{\mu\nu}
}
is the stress-energy of the elastic body with
\eqn{rhodef}{
\rho = \rho(f,H) \qquad (H^{IJ} := g^{\mu\nu}\del{\mu}f^I\del{\nu}f^J )
}
defining the proper energy density of the elastic body. By definition, $\rho$ is non-zero inside $W_T$ and vanishes outside. Letting $\Gamma_T$ denote the space-like boundary of $W_T$, the elastic field must also satisfy the boundary conditions
\leqn{EinElasB}{
n^\mu T_{\mu\nu} = 0 \text{\hspace{0.3cm} in $\Gamma_T$,}
}
where $n^\mu$ denotes the outward pointing unit normal to $\Gamma_T$. Initial conditions for \eqref{EinElasA.1}-\eqref{EinElasA.2} are given
by
\lalin{EinElasC}{
(g_{\mu\nu},\Lc_t g_{\mu\nu}) &= (g^0_{\mu\nu},g^1_{\mu\nu}) \text{\hspace{0.3cm} in $\Sigma$,} \label{EinElasC.1}  \\
(f^I,\Lc_t f^I) &= (f_0^I,f_1^I) \text{\hspace{0.6cm} in $\Sigma\cap W_T$,} \label{EinElasC.2}
}
where $\Sigma$ forms the ``bottom'' of the spacetime slab $M_T\cong [0,T)\times \Sigma$,
$\Sigma\cap W_T$ forms the bottom of the body world tube $W_T\cong [0,T)\times \Omega$, $t=t^\mu\del{\mu}$ is a future pointing
time-like vector field tangent to $\Gamma_T$ and normal to $\Sigma$, and the initial data satisfies the \emph{constraint equations}
\leqn{EinElasD}{
t_{\mu} G^{\mu\nu} = 2\kappa t_{\mu} T^{\mu\nu} \text{\hspace{0.3cm} in $\Sigma$.}
}

Together, the field equations
\eqref{EinElasA.1}-\eqref{EinElasA.2}, the boundary conditions \eqref{EinElasB}, and
the constraints \eqref{EinElasD} for the initial data
\eqref{EinElasC.1}-\eqref{EinElasC.2} make up the fundamental initial boundary value problem (IBVP) for
self-gravitating, relativistic elastic bodies.

\subsect{redEinElas}{The reduced Einstein-Elastic equations}
The method we use to solve the fundamental IBVP given by \eqref{EinElasA.1}, \eqref{EinElasA.2}, \eqref{EinElasB},
\eqref{EinElasC.1} and \eqref{EinElasC.2} is based on the well known technique of employing harmonic coordinates due originally to Y. Choquet-Bruhat, see \cite[Ch. VI, \S 7]{ChoquetBruhat:2009} for a general discussion.
This technique allows us to
replace the full Einstein equations \eqref{EinElasA.1} with the \emph{reduced} equations
given by
\eqn{EinElasE}{
R_{\mu\nu} - \nabla_{(\mu}\zeta_{\nu)} = 2\kappa\bigl( T_{\mu\nu}-\Half T g_{\mu\nu}\bigr),
}
where
\eqn{zetadef}{
\zeta^\gamma = g^{\mu\nu}\Gamma_{\mu\nu}^\gamma.
}
For this method to work, we must choose the initial data so that the additional constraint
\leqn{xidata}{
\zeta^\mu = 0 \quad \text{in $\Sigma$}
}
is satisfied.

The system of primary interest thus becomes the
\emph{reduced Einstein-Elastic
equations} given by
\lalin{redEinElasA}{
R_{\mu\nu} - \nabla_{(\mu}\zeta_{\nu)} &= 2\kappa\bigl( T_{\mu\nu}-\Half T g_{\mu\nu}\bigr)
\text{\hspace{0.3cm} in $M_T$,} \label{redEinElasA.1}  \\
\nabla_{\mu} T^{\mu\nu} &= 0 \text{\hspace{2.8cm} in $W_T$,} \label{redEinElasA.2}
}
where the inital data \eqref{EinElasC.1}-\eqref{EinElasC.2} satisfy the constraints
\eqref{EinElasD} and \eqref{xidata}. Our first step is then to establish the local existence
of solutions to the reduced equations \eqref{redEinElasA.1}-\eqref{redEinElasA.2}. Once this is accomplished, the final
step is verify that these solutions also satisfy
the full equations \eqref{EinElasA.1}-\eqref{EinElasA.2}. This strategy of first solving the reduced equations in order to obtain solutions to
the full equations is well known.
For details, at
least for sufficiently regular spacetimes, see  \cite[Ch. VI, \S 8]{ChoquetBruhat:2009}.
However, due to jump discontinuity in the stress-energy tensor at the matter-vacuum interface, we
cannot apply the standard results directly in our setting.
Instead, we find a suitable
weak formulation and use this to establish the equivalence between
the reduced and full systems of equations.  A complete proof of this equivalence
is given below in Section \ref{rlemat:unique}.

\subsect{matpic}{Material representation}
The boundary $\Gamma_T$, which defines the matter-vacuum interface, is a free boundary.
To fix this free boundary, we follow the standard method in elasticity and
employ the \emph{material representation} defined by
the map\footnote{We always use lower case Roman indices, e.g. $i,j,k$, to label the spatial
coordinate indices that run from $1$ to $n$.}
\eqn{phiA}{
x^i = \phi^i(X^0,X^I), \qquad i=1,2,\ldots,n,
}
that is uniquely determined by the requirement
\eqn{phiB}{
f^I(X^0,\phi(X^0,X^I)) = X^I, \qquad \forall \; (X^0,X^I)\in [0,T) \times \Omega.
}
For an extended discussion of the material representation in relativistic elasticity,
see \cite[\S 3.2]{Wernig:2006}.

Our regularity assumptions on the map $\phi^i$ will be that
\leqn{phiregA}{
\phi^i \in Y^{s+1}_T(\Omega),  \qquad s\in \Zbb_{>n/2+1},
}
where the spaces $Y_T(\Omega)$ are defined below in Section \ref{funct}.
Letting
\eqn{phiD}{
x^0 = \phi^0(X^0,X^I) := X^0,
}
we define the extended map\footnote{Upper case Greek, e.g. $\Lambda,\Gamma,\Sigma$,
indices will run from $0$ to $n$ and are used to label the material spacetime coordinates.}
\leqn{phiC}{
\phi : [0,T) \times \Omega \longrightarrow W_T \: :\: \Xv=(X^\Lambda) \longmapsto (\phi^\mu(X))=(X^0,\phi^i(X)).
}
A basic requirement for the material representation is that this map defines a diffeomorphism between the \textit{material cylinder} $[0,T) \times \Omega$ and the world tube $W_T$. Assuming this,
it follows that there must exist a positive constant $\gamma_0>0$ such that
\leqn{phiregB}{
\det\bigl(\partial_I \phi^j(\Xv)\bigr) > \gamma_0 > 0,
\qquad \forall \; \Xv\in [0,T)\times \Omega.
}

In the material representation, the components $\phi^i$ completely determine the elastic field, and
they are defined on the time-independent material cylinder. To describe the gravitational field
in the material representation, we need to extend the domain of the material map \eqref{phiC} to a neighborhood
of the material cylinder. We accomplish this by first fixing a (non-unique) total extension operator
\eqn{Edefa}{
\Ebb_\Omega \: : \: H^k(\Omega)\longrightarrow H^k(\Rbb^{n}), \qquad k \in \Zbb_{\geq 0},
}
which satisfies
\leqn{EdefBa}{
\Ebb_\Omega(u)|_{\Omega} = u, \AND
\norm{\Ebb_\Omega(u)}_{H^{k}(\Rbb^n)} \leq K\norm{u}_{H^k(\Omega)}
}
for some constant $K=K(k,n)>0$ independent of $u$. The existence of such an operator is established in
\cite{AdamsFournier:2003}; see Theorems 5.21 and 5.22, and Remarks 5.23 for details.
We use the extension operator to extend
$\phi$ via the prescription
\leqn{phitdefA}{
\phit^0(X^0,X^I) = X^0 \AND
\phit^i(X^0,X^I) = \Ebb_\Omega(\phi^i|_{\{X^0\}\times \Omega})(X^I)
}
for all $(X^0,X^I)\in [0,T)\times \Rbb^n$.
\begin{lem} \label{difflem}
There exists
an open neighborhood $\Nc \subset \Rbb^n$ of $\Omega$ for which
\leqn{phitdefB}{
\phit|_{[0,T)\times \Nc} \: : \: [0,T)\times \Nc \longrightarrow  \Mt_T:=\phit([0,T)\times \Nc)\subset M_T
}
is a diffeomorphism.
\end{lem}
\begin{proof}
Since $s>n/2+1$, it is clear from \eqref{phiregA}, \eqref{phiregB}, \eqref{EdefBa}, and Sobolev's inequality, see Theorem \ref{Sobolev},
that there exists an open neighborhood $\Nc$ of $\Omega$ such that
\eqn{difflem1}{
\det\bigl(\partial_I \phit^j(\Xv)\bigr) > \frac{\gamma_0}{2} > 0,
\qquad \forall \; \Xv\in [0,T)\times \Nc.
}
By the Implicit Function Theorem, we conclude that the map \eqref{phitdefB} is
a local diffeomorphism. Since $\Omega$ is bounded and $\phit$ is one-to-one on $[0,T)\times \Omega$, it is not difficult
to see that we
can, by shrinking $\Nc$ if necessary, arrange that $\phit$ remains one-to-one on $[0,T)\times \Nc$.
\end{proof}
\begin{rem} \label{nbhdrem}
Away from the boundary $\Gamma_T$ in the vacuum region $M_T\setminus W_T$, the local existence and uniqueness of solutions
to the vacuum Einstein equations is well known; for example, see \cite[Ch. VI, \S 8]{ChoquetBruhat:2009}.
Because of this, we can appeal to the finite propagation speed property satisfied by the vacuum Einstein equation,
and conclude that, as far as questions of local existence and uniqueness are concerned, we
lose no generality
 in restricting our attention to the open neighborhood $\Nc$ of the material manifold. By suitably choosing our extension operator $\Ebb_\Omega$,
we can, in fact, take $\Nc$ to be the open box, and by subsequently identifying
opposite sides, we may, without loss of generality, assume that
\eqn{nbhdrem1}{
\Nc \cong \Sigma \cong \Tbb^3
\AND
 \Mt_T = M_T.
}
The extension operator then maps
\eqn{Edef}{
\Ebb_\Omega \: : \: H^k(\Omega)\longrightarrow H^k(\Tbb^{n}), \qquad k \in \Zbb_{\geq 0},
}
and satisfies
\leqn{EdefB}{
\Ebb_\Omega(u)|_{\Omega} = u, \AND
\norm{\Ebb_\Omega(u)}_{H^{k}(\Tbb^n)} \leq K\norm{u}_{H^k(\Omega)}
}
for some constant $K=K(k,n)>0$ independent of $u$.
\end{rem}

We use the map \eqref{phitdefB} to pull back the components of the metric $g_{\mu\nu}$ as scalars
to get the \emph{material representation for the spacetime metric}
denoted by
\eqn{gammadef}{
\gt_{\mu\nu}(\Xv) := g_{\mu\nu}\circ \phit(\Xv).
}
For use below, we let
\eqn{Jdef}{
J = (J^\mu_\Lambda) := (\del{\Lambda}\phi^\mu) \AND  \Jt = (\Jt^\mu_\Lambda) := (\del{\Lambda}\phit^\mu)
}
denote the Jacobian matrices of $\phi$ and its extension $\phit$, respectively, and use
\eqn{Jchdef}{
\Jch = \bigl(\Jch^\Lambda_\mu\bigr) := \bigl((J^{-1})^\Lambda_\mu\bigr) \AND
\Jtch = \bigl(\Jtch^\Lambda_\mu\bigr) := \bigl((\Jt^{-1})^\Lambda_\mu\bigr)
}
to denote the inverses.  By definition of the extension, these Jacobian matrices satisfy
\eqn{Jrest}{
\Jt|_{[0,T) \times \Omega} = J
\AND  \Jtch|_{[0,T) \times \Omega} = \Jch.
}

As shown in Section 3.2, 4.1.2, and  4.6.2 \cite{Wernig:2006}, the
elastic field equations \eqref{EinElasA.2} and boundary conditions \eqref{EinElasB}, when expressed in the material representation,
are given by
\lalin{matElasA}{
\del{\Lambda}(L^\Lambda_i(\Xv,\gt,\phiv,\del{}\phiv)) &= w_i(\Xv,\gt,\phiv,\del{}\phiv) \text{\hspace{0.5cm} in $[0,T) \times \Omega$,} \label{matElasA.1}
\\
\nu_{\Lambda}(L^\Lambda_i(\Xv,\gt,\phiv,\del{}\phiv)) &= 0 \text{\hspace{2.625cm} in $[0,T)\times \del{}\Omega$,}
\label{matElasA.2}
}
where  $\nu_\Lambda = \delta_\Lambda^I \nu_I$ with $\nu_I$ the unit\footnote{This is with respect to the Euclidean inner-product.} conormal
to $\del{}\Omega$. Here, we assume that
\begin{enumerate}
\item[(a)]
$L^\Lambda(\Xv,\gt,\del{}\phiv)$ and
$ w_i(\Xv,\gt,\phiv,\del{}\phiv)$ are smooth for\footnote{As discussed
in \cite{Wernig:2006}, the maps $L^\Lambda_i(\Xv,\gt,\phiv,\del{}\phiv)$ and $w_i(\Xv,\gt,\phiv,\del{}\phiv)$
originate from a Lagrangian
$L(\Xv,\gt,\phi,\del{}\phiv)$ according to $L^\Lambda_i = \frac{\del{}L\;}{\del{}(\del{\Lambda}\phi^i)\;}$
and $w_i = \frac{\del{}L\;}{\del{}\phi^i\;}$,
which make it clear that the field equations \eqref{matElasA.1} are nothing more that the Euler-Lagrange equations
for $L$.}
\eqn{Lsmooth}{
(\Xv,\gt,\del{}\phiv)=(X^\Lambda,\gt_{\mu\nu},\phi^i,\del{\Lambda}\phi^i) \in ([0,T) \times \Omega)\times
\Uct\times \Vct \times \Wct
}
with\footnote{Here, we are using $\Sbb{n+1}$ to denote the sets of symmetric, $(n+1)\times(n+1)$-matrices.}
\eqn{Uctdef}{
\Uct = \{\, \gt \in \Sbb{n+1}\, |\, \det(\gt) < 0\, \},
}
$\Vct$ an open set in $\Rbb^{n}$,
and
\eqn{Wctdef}{
\Wct\subset \bigl\{\, (J_\Lambda^j)\in \Rbb^{(n+1)\times n}\, \bigl| \, \det\bigl(J_I^j\bigr)
 > 0 \, \bigr\}
}
an open set.
\end{enumerate}
From \cite{Wernig:2006}, we know that the elasticity tensor
\eqn{Ltnedef}{
L^{\Lambda\Gamma}_{ij}(\Xv,\gt,\del{}\phiv) := \frac{\del{}L^\Gamma_j((\Xv,\gt,\del{}\phiv))}{\del{}(\del{\Lambda}\phi^i)\;}
}
satisfies the symmetry condition\footnote{This symmetry follows automatically from the fact
the the field equations are derived from a Lagragian; since $L^\Lambda_i = \frac{L\;}{\del{}(\del{\Lambda}\phi^i)\;}$
implies that
$L^{\Lambda\Gamma}_{ij} = \frac{\del{}^2 L\;}{\del{}(\del{\Lambda}\phi^i)\del{}(\del{\Gamma}\phi^j)\; }$,
the symmetry condition \eqref{Lsym} follows from the commutativity of mixed partial derivatives.}
\leqn{Lsym}{
L^{\Lambda\Gamma}_{ij} = L^{\Gamma\Lambda}_{ji}.
}

For our existence result, we need restrict ourselves to elastic materials whose
elasticity tensors satisfy the following two conditions:
\begin{enumerate}
\item[(b)] there exists open sets
\eqn{UctWct}{
\Uc \subset \overline{\Uc} \subset \Uct,\quad \Vc \subset \overline{\Vc} \subset \Vct \AND
\Wc \subset \overline{\Wc} \subset \Wct
}
and a positive constant $\kappa_0>0$ such that
\eqn{L00}{
\xi^i L^{00}_{ij}(\Xv,\gt,\phiv,\del{}\phiv) \xi^j \leq  -\kappa_0 |\xi|^2
}
for all $(\Xv,\gt,\phiv,\del{}\phiv,\xi)$
$\in$ $[0,T)$ $\times$ $\Omega$ $\times$ $\Uc$ $\times$ $\Vc$$\times$ $\Wc$ $\times$ $\Rbb^n$, and
\item[(c)] there exist constants $\kappa_1>0$ and $\gamma_1 \in \Rbb$ for which the spatial components $L^{IJ} = (L^{IJ}_{ij})$
of the elasticity tensor satisfy
the coercitivity condition:
\eqn{coercA}{
\ip{\del{I}\zeta}{L^{IJ}((X^0,\cdot),\gt(X^0),\phiv(X^0),\del{}\phiv(X^0))\del{J} \zeta}_{L^2(\Omega)}
\geq \kappa_1 \norm{\zeta}^2_{H^1(\Omega)} - \gamma_1 \norm{\zeta}^2_{L^2(\Omega)}
}
for each $X^0 \in [0,T)$, $\zeta\in C^1(\overline{\Omega},\Rbb^n)$, and
\eqn{coercB}{
(\gt,\phiv) \in C^0([0,T]\times \overline{\Omega},\Sbb{n+1})\times  \in C^1([0,T]\times \overline{\Omega},\Rbb^n)
}
satisfying
\eqn{coercC}{
\bigl(\gt(\Xv),\phi(\Xv),\del{}\phi(\Xv)\bigr)\in \Uc\times\Vc\times \Wc \quad \forall\; \Xv \in [0,T)\times \Omega.
}
\end{enumerate}

With the material representation for the elastic field equations complete, we now
turn to expressing the reduced Einstein equations \eqref{EinElasA.1} in the material representation. We begin
this process by recalling the well known expansion for the
left hand side of \eqref{EinElasA.1} given by
\leqn{redRicci}{
-2R_{\mu\nu}+2\nabla_{(\mu}\zeta_{\nu)} = \frac{1}{\sqrt{|g|}}\del{\alpha}
\Bigl[\bigl(\sqrt{|g|}g^{\alpha\beta}
\del{\beta}g_{\mu\nu}\bigr) - Q_{\mu\nu}(g,\del{}g) \Bigr],
}
where $Q_{\mu\nu}(g,\del{}g)$ is analytic for $(g,\del{}g)\in \Uct\times \Sbb{n+1}$
and quadratic in $\del{}g$. From the change of variables formula from
multivariable calculus, we observe that
\leqn{change}{
(\del{\alpha}g_{\mu\nu})\circ \phit = \Jtch^\Lambda_\alpha \del{\Lambda}\gt_{\mu\nu}.
}
Next, we recall the transformation law for vector fields $Y=Y^\alpha\del{\alpha}$ given by:
\eqn{vecchange}{
\frac{1}{\sqrt{|\gb|}}\del{\Lambda}\bigl(\sqrt{|\gb|}\Yb^\Lambda\bigr) =
\left(\frac{1}{\sqrt{|g|}}\del{\alpha}\bigl(\sqrt{|g|}Y^\alpha\bigr) \right)\circ \phit
}
where $\Yb^\Lambda = (\phit^*Y)^\Lambda = \Jtch^\Lambda_\alpha Y^\alpha \circ \phit$,
and $|\gb| = |\phit^* g| = \det(\Jt)^2 |\gt|$ with $|\gt|=-\det(\gt_{\alpha\beta})$.
Setting $Y^\alpha = g^{\alpha\beta}\del{\beta}g_{\mu\nu}$ in this formula, we see,
with the help of \eqref{change}, that
\leqn{divdelg}{
\biggl(\frac{1}{\sqrt{|g|}}\del{\alpha}
\bigl(\sqrt{|g|}g^{\alpha\beta}
\del{\beta}g_{\mu\nu}\bigr)\biggr)\circ \phit = \frac{1}{\det(\Jt)\sqrt{|\gt|}}\del{\Lambda}
\bigl(A^{\Lambda\Gamma}\del{\Gamma} \gt_{\mu\nu}\bigr),
}
where
\eqn{Adef}{
A^{\Lambda\Gamma} = A^{\Lambda\Gamma}(\Jt,\gt) := \det(\Jt)\Jtch^\Lambda_\alpha\sqrt{|\gt|}\gt^{\alpha\beta}\Jtch^\Gamma_\beta
\AND
\gt^{\alpha\beta} = \gt^{\alpha\beta}(\gt) :=(\gt_{\alpha\beta})^{-1}.
}
We assume that
\begin{enumerate}
\item[(d)] there exists a constant $\kappa_2>0$ such
\eqn{lambda00}{
\gt^{00}(\gt) \geq \kappa_2, \qquad \forall\, \gt \in \Uc,
}
\end{enumerate}


Taken together, \eqref{redRicci}, \eqref{change} and \eqref{divdelg} show that
\leqn{matredEinA}{
\del{\Lambda}\bigl(A^{\Lambda\Gamma}(\Jt,\gt)\del{\Gamma} \gt_{\mu\nu}\bigr) =
\det(\Jt)Q_{\mu\nu}(\gt,\Jtch\del{} \gt) + \chi_\Omega\Tc_{\mu\nu},
}
where
\eqn{TcdefA}{
\Tc_{\mu\nu} = - 4\kappa  \det(J)\sqrt{|\gt|}\bigl(T_{\mu\nu}\circ \phi -\Half \gt^{\alpha\beta}
T_{\alpha\beta}\circ \phi \gt_{\mu\nu} \bigr),
}
is the material representation of the reduced Einstein equations \eqref{EinElasA.1}. The characteristic
function $\chi_\Omega$ of the set $\Omega$ has been included as a prefactor to the stress
energy contributions on the right hand side of \eqref{matredEinA}  to
highlight the jump discontinuity
across the matter-vacuum interface, defined by $[0,T) \times \del{} \Omega$, and to make the
vanishing of stress energy tensor vanish in the vacuum region $[0,T)\times (\Nc\setminus \!\Omega)$
crystal clear. By assumption (a) above,
\eqn{TcdefB}{
\Tc_{\mu\nu} = \Tc_{\mu\nu}(\Xv,\gt,\phiv,\del{}\phiv)
}
is smooth for $(\Xv,\gt,\phiv,\del{}\phiv))$ $\in$ $[0,T)$ $\times$ $\Omega$ $\times$
$\Uct$ $\times$ $\Vct$ $\times$ $\Wct$.

Summarizing the above results, the complete IBVP for the reduced Einstein-Elastic problem in the material
representation is given by
\lalin{matEinElasA}{
\del{\Lambda}\bigl(A^{\Lambda\Gamma}(\Jt,\gt)\del{\Gamma} \gt_{\mu\nu}\bigr) &=
\det(\Jt)Q_{\mu\nu}(\gt,\Jtch\del{} \gt) + \chi_\Omega\Tc_{\mu\nu}(\Xv,\gt,\phiv,\del{}\phiv)
 \text{\hspace{0.3cm} in $[0,T)\times \Tbb^n$,} \label{matEinElasA.1} \\
 (\gt_{\mu\nu},\del{0}\gt_{\mu\nu}) &= (\lambdah^0_{\mu\nu},\lambdah^1_{\mu\nu})
 \text{\hspace{5.05cm} in $\{0\}\times \Tbb^n$,} \label{matEinElasA.2}\\
\del{\Lambda}(L^\Lambda_i(\Xv,\gt,\phiv,\del{}\phiv)) &= w_i(\Xv,\gt,\phiv,\del{}\phiv) \text{\hspace{4.3cm} in $[0,T) \times \Omega$,}
\label{matEinElasA.3}\\
\nu_{\Lambda}L^\Lambda_i(\Xv,\gt,\phiv \del{}\phiv) &= 0 \text{\hspace{6.425cm} in $[0,T) \times \del{}\Omega$,}
\label{matEinElasA.4}\\
(\phi^i,\del{0}\phi^i) &= (\phih^i_0,\phih^i_1)   \text{\hspace{5.45cm} in $\{0\}\times\Omega$.}
\label{matEinElasA.5}
}
We further assume that the initial data
\begin{enumerate}
\item[(e)]
lies in the spaces\footnote{All function spaces are described in Section \ref{funct}.}
\eqn{matEinElasidata}{
(\lambdah^0_{\mu\nu},\lambdah^1_{\mu\nu},\phih^i_0,\phih^i_1) \in \Hc^{2,s+1}(\Tbb^n)\times \Hc^{2,s}(\Tbb^n)
\times H^{s+1}(\Omega)\times H^{s}(\Omega)
}
for $s\in \Zbb_{>n/2+1}$,
\item[(f)] is chosen so that the map
\eqn{phihmapa}{
\Tbb^n \ni X \longmapsto (\Ebb_\Omega(\phih^i_0)(X))\in \Tbb^n
}
is a $H^{s+1}(\Tbb^n)$-diffeomorphism, and there exists a constant $\gamma_0 >0$ such that
\eqn{phihmapb}{
\det\bigl(\partial_J \phih^i_0(X)\bigr) > \gamma_0 > 0,
\qquad \forall \; X\in \Tbb^n,
}
\item[(g)] satisfies the \emph{compatibility conditions} given by
\lgath{EEcompat}{
\del{0}^\ell \! \gt_{\mu\nu} \bigl|_{X^0=0} \in \Hc^{m_{s+1-\ell},s+1-\ell}(\Tbb^n)
\quad \ell=2,\ldots,s+1, \label{EEcompat.1} \\
\del{0}^\ell \!\bigl(\nu_\Lambda L^\Lambda_i(\Xv,\gt,\del{}\phiv)\bigr)\Bigr|_{X^0=0} \in H^{s-\ell}(\Omega)\cap H^1_0(\Omega),
 \quad \ell=0,1,\ldots,s-1,  \label{EEcompat.2}}
where
\leqn{mEEdef}{
m_j = \begin{cases} 2 & \text{if $j\geq 2$} \\
j & \text{otherwise} \end{cases},
}
and
\item[(h)] satisfies, after transforming to the Eulerian representation,
the constraint equations \eqref{EinElasD} and \eqref{xidata}.
\end{enumerate}

\begin{rem} \label{idatarem}
\begin{enumerate}
$\;$

\item[(i)] We do not actually need that $\Sigma \cong \Tbb^n$ and that the initial data satisfies the constraint equations
on all of $\Tbb^n$. We make this assumption in order to focuss on the essential elements of the proof.
It is not difficult, using domain of dependence arguments, to
see that it is enough to assume the $\Sigma$ is open in $\Tbb^n$, and that the constraint equations
are satisfied only on $\Sigma$. With this change, all of the arguments used below would go through
with $[0,T)\times \Tbb^n$ replaced by an appropriate lens shaped domain having $\Sigma$ as its base.
\item[(ii)] The problem of classifying initial data that satisfies assumptions (e)-(h) for general
initial hypersurfaces $\Sigma$ is a very difficult problem. Even in the simpler setting of vacuum spacetimes, our
understanding of the solution space for the constraint equations is far from complete. However, in
special situations, we do know that the set of initial data satisfying (e)-(h) is non-empty. For example,
it is clear that stationary solutions of the Einstein-Elastic system, e.g. \cite{Andersson_et_al:2008,Andersson_et_al:2010}, automatically satisfy (e)-(h). Presumably, it would be possible
to use an Implict Function Theorem argument to construct an open neighborhood of the stationary initial data
satisfying (e)-(h). We also remark the initial data constructed in \cite{Andersson_et_al:2011}, which satisfies the Newtonian analogue of assumptions (e)-(h),
can be perturbed to fully relativistic initial data satisfying (e)-(h) using the method of
\cite{Lottermoser:1992}, see also \cite{Oliynyk:CMP_2007}. Details of this construction and
related investigations will be presented elsewhere.
\end{enumerate}
\end{rem}

\sect{locexist}{Local existence theorems}

We are now ready to state our main results, which are contained in the two following theorems.
The first theorem establishes local existence and uniqueness of
solutions to the reduced Einstein-Elastic system, while the second shows that these solutions also satisfy the full Einstein-Elastic system provided that
the initial data satisfies the constraints \eqref{EinElasD} and \eqref{xidata}. Proofs of these theorems are given in Sections \ref{rlemat} and \ref{flemat}.
\begin{thm} \label{redlocthm}
Suppose the assumptions (a) to (g) from Section \ref{matpic} hold. Then there exists a $T>0$ and a unique solution $(\gt_{\mu\nu},\phi^i)\in X^{s+1}_T(\Tbb^n)\times Y^{s+1}_T(\Omega)$
on $[0,T)\times\Omega$
of the IBVP \eqref{matEinElasA.1}-\eqref{matEinElasA.5}.
\end{thm}

\begin{rem} \label{redlocthmrem}
Although we do not need it here, it is worthwhile noting that arguments from \cite{AnderssonOliynyk:2014,Koch:1993}  can
be adapted to establish a continuation principle for the solutions from Theorem \ref{redlocthm}. Specifically, it can
be shown that if the solution $(\gt_{\mu\nu},\phi^i)$ from Theorem \ref{redlocthm} continues to lie in the interior of the regions where assumptions (a) to (d)
are satisfied for some $\kappa_0,\kappa_1,\kappa_2>0$ and $\gamma_1\in \Rbb$,
and $\norm{\gt}_{W^{1,\infty}([0,T)\times\Omega)}+\norm{\phiv}_{W^{2,\infty}([0,T)\times\Omega)}<\infty$, then there
exists a $T^*>T$ and a unique extension $(\gt^*_{\mu\nu},\phi^i_*) \in X^{s+1}_{T^*}(\Tbb^n)\times Y^{s+1}_{T^*}(\Omega)$
of $(\gt_{\mu\nu},\phi^i)$ that solves the IVBP \eqref{matEinElasA.1}-\eqref{matEinElasA.5} on $[0,T^*)\times \Omega$.
\end{rem}

\begin{thm} \label{fulllocthm}
Suppose that assumption (a) to (g) from Section \ref{matpic} hold, and let $(\gt_{\mu\nu},\phi^i)\in X^{s+1}_T(\Tbb^n)\times Y^{s+1}_T(\Omega)$ be the solutions
to the reduced Einstein-Elastic system from Theorem \ref{redlocthm}. If assumption (h) from Section \ref{matpic} also holds, then
$\zeta_\nu\circ \phit$ $=$ $0$ on $[0,T)$$\times$$\Omega$ and the pair $(\gt_{\mu\nu},\phi^i)$ determine a solution of the full Einstein-Elastic system
on $[0,T)$$\times$ $\Omega$ in the material representation.
\end{thm}

\begin{rem} \label{fulllocrem}
Using the techniques develop in this article, it is not difficult to adapt the arguments
from \cite[Ch. VI, \S 8.3]{ChoquetBruhat:2009} to show that the solution from Theorem \ref{fulllocthm}
are geometrically unique in the sense of Theorem 8.4 from  \cite[Ch. VI, \S 8.3]{ChoquetBruhat:2009}.
\end{rem}

\sect{wave}{Wave equations}

The proofs of Theorems \ref{redlocthmrem} and \ref{fulllocthm} rely on existence and uniqueness theorems for two classes of wave equations.
The first class concerns linear wave equations involving a jump discontinuity in the source
term, and the existence and uniqueness result presented in
Theorem \ref{tlinthm} below for these systems
is an adaptation of Theorem 3.4 from \cite{AnderssonOliynyk:2014}. Since the proof is
similar, we only record the
essential differences and leave remainder of the details to the interested reader.
The second class of wave equations
consist of
non-linear systems in divergence form on bounded domains with Neumann boundary conditions. The existence and uniqueness results for this class of equations, presented in Section \ref{divwave}, is a variant of Theorem 1.1  of \cite{Koch:1993}.


\subsect{prelim}{Preliminaries} Before proceeding with the proof of Theorem \ref{tlinthm},
we first fix our notation and conventions that will be used throughout this section.

In the following,
we use $(x^{\mu})_{\mu=0}^n$  to denote Cartesian coordinates
on $\Rbb^{n+1}$; we use  $x^0$ and $t$, interchangeability, to denote the time coordinate,
and $(x^i)_{i=1}^n$ to denote the spatial coordinates. We also use $x=(x^1,\ldots,x^n)$
and $\xv = (x^0,\ldots,x^n)$ to denote spatial and spacetime points, respectively.
Partial derivatives are denoted as above, see \eqref{pardef},
and
we use $Du(x) = (\del{1}u(x),\ldots,\del{n}u(x))$ and
$\del{}u(\xv) =  (\del{0}u(\xv),Du(\xv))$ to denote the spatial and spacetime gradients, respectively.

\subsubsect{sets}{Sets}
The following subset of $\Rbb^n$ will be of interest:
\alin{setsdef}{
Q^+_1 &= \{\: (x^1,\ldots,x^n) \: |\: -1 < x^1,x^2,\ldots,x^{n-1} < 1, \quad 0 < x^n < 1 \: \} \\
\intertext{and}
Q_\delta &= \{\: (x^1,\ldots,x^n) \: |\: -1 \leq  x^1,x^2,\ldots,x^{n} \leq 1  \:\}.
}
We will also need to identify the opposite sides of the n-box $Q_1$ so that\footnote{Here, $\sim$ denotes the equivalence relation on
$Q_1$ determined by the identification of the opposite sides of the boundary.}
\eqn{Tbbdefa}{
 Q_1/{\sim} {} \approx{} \Tbb^n .
}
We note that  under this identification, the Carestian coordinates $x=(x^i)$ on $\Rbb^n$  define periodic coordinates on
$\Tbb^n$.
The following open and connected subset of $\Tbb^n$ with smooth boundary will also be of interest
\eqn{Omegadef}{
\Omega_1  =  Q^+_{1}/{\sim}.
}

Given an open set $\Omega$ of $\Tbb^n$ with smooth boundary,
we let $\chi_\Omega$ denote the characteristic function, and we use $\Omega^c$ to
denote the interior of its complement, that is
\eqn{Omeagcdef}{
\Omega^c := \Tbb^n \setminus \overline{\Omega}.
}

\subsubsect{funct}{Function spaces} $\;$

\bigskip

\noindent\textit{Spatial function spaces}

\bigskip

Given an open set $\Omega \subset \Tbb^n$ with smooth boundary, we use the standard notation $H^s(\Omega)$, $s\in \Zbb_{\geq 0}$, to denote
the $L^2$ Sobolev spaces. We also define the intersection spaces
\eqn{HksdefT}{
\Hc^{k,s}(\Tbb^n) =  H^k(\Tbb^n)\cap  H^s(\Omega)\cap H^s(\Omega^c)  \quad (s\geq k; k,s\in\Zbb_{\geq 0}),
}
which we equip with the norm
\eqn{HksnormT}{
\norm{u}_{\Hc^{k,s}(\Tbb^n)}^2 = \norm{u}_{H^s(\Omega)}^2 + \norm{u}_{ H^k(\Tbb^n)}^2 + \norm{ u}^2_{H^s(\Omega^c)}.
}

\bigskip

\noindent\textit{Spacetime function spaces}

\bigskip

Given an open subset $\Omega \subset \Tbb^n$ with smooth boundary and a $T>0$, we define the spaces
\leqn{XTdef}{
X^s_T(\Tbb^n) = \bigcap_{\ell=0}^s W^{\ell,\infty}\bigl([0,T),\Hc^{m_{s-\ell},s-\ell}(\Tbb^n)\bigr),
}
where $m_\ell$ is a defined above by \eqref{mEEdef},
\leqn{XcTdef}{
\Xc^s_T(\Tbb^n) = \bigcap_{\ell=0}^s W^{\ell,\infty}\bigl([0,T),\Hc^{0,s-\ell}(\Tbb^n)\bigr)
}
and
\leqn{YTdef}{
Y^s_T(\Omega) =  \bigcap_{\ell=0}^s W^{\ell,\infty}\bigl([0,T),H^{s-\ell}(\Omega)\bigr).
}

We also define the  \emph{energy norms}:
\gath{XTnormA}{
\norm{u}_{E^s(\Tbb^n)}^2 = \sum_{\ell=0}^s \norm{\del{t}^\ell u}^2_{\Hc^{m_{s-\ell},s-\ell}(\Tbb^n)}, \qquad
\norm{u}_{\Ec^s(\Tbb^n)}^2 = \sum_{\ell=0}^s \norm{\del{t}^\ell u}^2_{\Hc^{0,s-\ell}(\Tbb^n)}, \\
\norm{u}^2_{E^s(\Omega)} =  \sum_{\ell=0}^s \norm{\del{t}^\ell u}^2_{H^{s-\ell}(\Omega)}, \qquad
\norm{u}_{E^{s,r}(\Tbb^n)}^2 = \sum_{\ell=0}^r \norm{\del{t}^\ell u}^2_{\Hc^{m_{s-\ell},s-\ell}(\Tbb^n)} \quad (s\geq r)
\intertext{and}
\norm{u}_{E(\Tbb^n)}^2 = \norm{u}_{E^1(\Tbb^n)}^2 = \norm{u}_{H^1(\Tbb^n)}^2 + \norm{\del{t}u}^2_{L^2(\Tbb^n)}.
}
In terms of these energy norms, we can write the norms of the spaces \eqref{XTdef}, \eqref{XcTdef} and \eqref{YTdef} as
\gath{XTnormB}{
\norm{u}_{X^s_T(\Tbb^n)} = \sup_{0\leq t < T} \norm{u(t)}_{E^s(\Tbb^n)}, \qquad \norm{u}_{\Xc^s_T(\Tbb^n)} = \sup_{0\leq t < T} \norm{u(t)}_{\Ec^s(\Tbb^n)}
\intertext{and}
\norm{u}_{Y^s_T(\Omega)} =  \sup_{0\leq t < T} \norm{u(t)}_{E^s(\Omega)},
}
respectively. Finally, we define the following subspaces of \eqref{XTdef} and \eqref{YTdef}: 
\alin{XVdef}{
CX^s_T(\Tbb^n) = \bigcap_{\ell=0}^s C^{\ell}\bigl([0,T),\Hc^{m_{s-\ell},s-\ell}(\Tbb^n)\bigr)
\AND
CY^s_T(\Omega) =  \bigcap_{\ell=0}^s C^{\ell}\bigl([0,T),H^{s-\ell}(\Omega)\bigr).
}

\subsubsect{cost}{Estimates and constants}

We employ that standard notation
\eqn{lesssimA}{
a \lesssim b
}
for inequalities of the form
\eqn{lesssimB}{
a \leq C b
}
in situations where the precise value or dependence on
other quantities of the constant $C$ is not required.  On the other hand,  when the dependence of the constant
on other inequalities needs to be specified, for example if the constant depends on the norms $\norm{u}_{L^\infty(\Tbb^n)}$ and $\norm{v}_{L^\infty(\Omega)}$, we use the notation
\eqn{lesssimC}{
C = C(\norm{u}_{L^\infty(\Tbb^n)},\norm{v}_{L^\infty(\Omega)}).
}
Constants of this type will always be non-negative, non-decreasing, continuous functions of their arguments.

\subsubsect{domain}{A simple extension operator}

Given an open set $\Omega$ in $\Tbb^n$, we define the trivial extension operator by
\eqn{chiext}{
\chi_\Omega u(x) = \begin{cases} u(x) & \text{if $x\in \Omega$} \\ 0 & \text{otherwise} \end{cases}.
}
Clearly, this defines a bounded linear operator from $L^p(\Omega)$ to $L^p(\Tbb^n)$.


\subsubsect{mollifier}{Smoothing operator}
We use  $S_\lambda$ ($\lambda >0$) to denote the smoothing operator on $\Tbb^n$ from
Section 2.6 of \cite{AnderssonOliynyk:2014}, and recall that $S_\lambda$
satisfies the following properties:

\begin{prop} \label{mollprop}
Suppose $\Omega$ is an open set in  $\Tbb^n$ with a
smooth boundary,  $1\leq p < \infty$ and $s \in \Zbb_{\geq 0}$. Then
there exists a family of continuous linear maps
\eqn{mollprop1}{
S_\lambda \: :\: W^{s,p}(\Tbb^n) \longrightarrow W^{s,p}(\Tbb^n) \quad \lambda > 0
}
satisfying
\gath{mollprop2}{
\norm{S_\lambda \chi_\Omega u}_{W^{k,p}(\Tbb^n)} < \infty, \quad k\geq s,\qquad
 \norm{S_\lambda \chi_\Omega u}_{W^{s,p}(\Omega)} \lesssim \norm{u}_{W^{s,p}(\Omega)}, \quad 0<\lambda \leq 1,
\intertext{and}
\lim_{\lambda\searrow 0} \norm{S_\lambda \chi_\Omega u - u}_{W^{s,p}(\Omega)}  = 0
}
for all $u\in W^{s,p}(\Omega)$. Moreover, the $S_\lambda$ are well-defined, continuous linear operators on $\Hc^{m,s}(\Tbb^n)$ and satisfy
\gath{mollcor1}{
\norm{S_\lambda  u}_{\Hc^{\ell,k}(\Tbb^n)} < \infty, \quad k\geq s,\; \ell \geq m, \; k\geq \ell, \qquad
 \norm{S_\lambda  u}_{\Hc^{m,s}(\Tbb^n)} \lesssim \norm{u}_{\Hc^{m,s}(\Omega)}, \quad 0<\lambda \leq 1,
\intertext{and}
\lim_{\lambda\searrow 0} \norm{S_\lambda  u - u}_{\Hc^{m,s}(\Tbb^n)}  = 0
}
for all  $u \in \Hc^{m,s}(\Tbb^n)$.
\end{prop}

\subsect{tlin}{Linear wave equations with a jump discontinuity}

The class of linear wave equations with a jump discontinuity that will be of interest to us are
\lalin{tlinIVP}{
\del{\mu}(A^{\mu\nu}\del{\nu} U) &=  F  +  \chi_\Omega H \quad \text{in $[0,T)\times \Tbb^n$,} \label{tlinIVP.1} \\
(U,\del{t} U)|_{t=0} &= (\Ut_0,\Ut_1)  \quad \text{ in $\{0\}\times\Tbb^n$,}
\label{tlinIVP.2}
}
where
\begin{enumerate}[(i)]
\item $U(t,x)$, $F(t,x)$ and $H(t,x)$ are vector valued in $\Rbb^N$ for some $N\in \Nbb$,
\item $\Omega$ is an open set in $\Tbb^n$ with smooth boundary,
\item $A^{\mu\nu}$ is of the form
\eqn{tlinIVPa}{
A^{\mu\nu} = \det(J)\Jch^\mu_\alpha a^{\alpha\beta} \Jch^{\nu}_\beta \qquad (a^{\alpha\beta}=a^{\beta\alpha}),
}
where
\eqn{Jdef}{
J = (J^\mu_\nu) := (\del{\nu}\Psi^\nu)
}
is the Jacobian matrix of a diffeomorphism
\eqn{Psidef}{
\Psi \: : \: [0,T)\times \Tbb^n \longrightarrow [0,T)\times \Tbb^n \: :\:
\xv \longrightarrow \Psi(\xv) = (\Psi^\nu(\xv))
}
and
\eqn{Jchdef}{
\Jch = J^{-1},
}
\item the initial data
\leqn{tlincompata}{
(\Ut_0,\Ut_1) \in \Hc^{2,s+1}(\Tbb^n)\times \Hc^{2,s}(\Tbb^n) \quad s\in \Zbb_{>n/2+1}
}
satisfies the \emph{compatibility conditions}
\leqn{tlincompat}{
\Ut_\ell := \del{t}^\ell U |_{t=0} \in \Hc^{m_{s+1-\ell},s+1-\ell}(\Tbb^n) \quad \ell=2,\ldots,s,
}
and,
\item there exist constants $\gamma,\kappa >0$  for which the matrix $a^{\mu\nu}$ and the Jacobian matrix $J^\mu_\nu$ satisfy
\lgath{tlinIVPb}{
\frac{1}{\gamma} |\xi|^2 \leq  a^{ij}(\xv)\xi_i \xi_j \leq \gamma |\xi|^2, \quad  a^{00}(\xv) \leq -\kappa, \label{tlinIVPb.1}\\
\frac{1}{\gamma} \leq \det{J(\xv)} \leq \gamma \AND  |J^\mu_\nu(\xv)| \leq \gamma \label{tlinIVPb.2}
}
for all $\xv \in [0,T)\times \Tbb^n$ and $\xi=(\xi_i)\in \Rbb^n$.
\end{enumerate}

\begin{thm} \label{tlinthm}
Suppose $n\geq 3$, $s\in \Zbb_{>n/2+1}$, $T>0$, $a^{\mu\nu}=a^{\nu\mu}\in \Xc_T^{s}(\Tbb^n)$, $\del{}a^{\mu\nu} \in \Xc_T^{s}(\Tbb^n)$,
 $F \in \Xc_T^{s}(\Tbb^n)$, $H \in Y^{s}_T(\Omega)$,
$(\Ut_0,\Ut_1) \in \Hc^{2,s+1}(\Tbb^n)\times \Hc^{2,s}(\Tbb^n)$
satisfy the compatibility conditions \eqref{tlincompat}, $\Psi^{\mu} \in \Xc_T^{s+1}(\Tbb^n)$,
$J^\mu_\nu=\del{\nu}\Psi^\mu \in \Xc_T^{s}(\Tbb^n)$, $DJ^\mu_\nu \in \Xc_T^{s-1}(\Tbb^n)$,
$a^{\mu\nu}$ and $J^\mu_\nu$ satisfy \eqref{tlinIVPb.1}-\eqref{tlinIVPb.2} for some $\kappa,\gamma > 0$,
and let
\gath{linGthm1}{
\alpha = \norm{J}_{\Xc_T^{s}(\Tbb^n)}+ \norm{DJ}_{\Xc_T^{s-1}(\Tbb^n)} +
\norm{a}_{\Xc_T^{s}(\Rbb)}+ \norm{Da}_{\Xc_T^{s-1}(\Tbb^n)}, \\
\sigma(t) = \bigl(1+\norm{\del{}a(t)}_{\Ec^s(\Tbb^n)}\bigr) \AND
\mu = \int_0^T \sigma(\tau)\, d\tau.
}
Then the IVP \eqref{tlinIVP.1}-\eqref{tlinIVP.2} has a unique solution
$U \in CX_T^{s+1}(\Tbb^n)$ that satisfies the energy estimate
\alin{linGthm2}{
\norm{U(t)}_{E^{s+1}(\Tbb^n)}  \leq C(\alpha,\mu,\gamma,\kappa)&
\biggl(\norm{U^\lambda(0)}_{E^{s+1}(\Tbb^n)}
+ \norm{F(0)}_{\Ec^{s-1}(\Tbb^n)}
   \\
+ \norm{H(0)}_{E^{s-1}(\Omega)}&
 +
\int_0^t \sigma(\tau)\norm{U^\lambda(\tau)}_{E^{s+1}(\Tbb^n)} + \norm{F(\tau)}_{\Ec^{s}(\Tbb^n)} +
\norm{H(\tau)}_{E^s(\Omega)} \, d\tau \biggr)
}
for $0\leq t < T$.
\end{thm}
\begin{proof}
The proof of this theorem follows from a modification of the proof of Theorem 3.4
from \cite{AnderssonOliynyk:2014}. We will only highlight the essential changes.
We begin by noting that it is clear from the proofs of Theorems 3.2 and 3.4 from \cite{AnderssonOliynyk:2014} that, for the
purpose of establishing existence and uniqueness of solutions to the
IVP  \eqref{tlinIVP.1}-\eqref{tlinIVP.2},
it is sufficient to consider a 1-parameter family of IVPs on the n-torus $\Tbb^n \cong Q_1/\!\!\sim$ of the form:
\lalin{tplinIVP}{
\del{\mu}(A^{\mu\nu}\del{\nu} U^\lambda) -\phi U^\lambda &=  S_\lambda(F)  +  S_\lambda(\chi_{\Omega_1}H) \text{
\hspace{0.2cm} in $[0,T)\times \Tbb^n$,} \label{tplinIVP.1} \\
(U^\lambda,\del{t} U^\lambda)|_{t=0} &= (S_\lambda\Ut_0,S_\lambda\Ut_1)   \text{\hspace{1.175cm} in $\Tbb^n$,}
\label{tplinIVP.2}
}
where $\lambda \in (0,1]$,
\begin{enumerate}[(a)]
\item $S_\lambda$ is the smoothing operator from Proposition \ref{mollprop},
\item the coefficients $A^{\mu\nu}$ are given by
\eqn{tplinIVPa}{
A^{\mu\nu} = \det(J)\Jch^\mu_\alpha a^{\alpha\beta} \Jch^\nu_\beta,
}
where
\gath{tplinIVPb}{
J^\mu_\nu = \del{\nu}\Psi^\mu_\lambda, \qquad
\Jch = J^{-1}, \qquad
\Psi^\mu(\xv) = x^\mu + \ep S_{\lambda}(\psi^\mu)(\xv), \\
a^{\alpha\beta} = m^{\alpha\beta} + \ep S_\lambda(b^{\mu\nu})
\AND
(m^{\mu\nu}) = \diag(-1,1,1,1),
}
\item the coefficients $\psi^{\mu}$, $b^{\mu\nu}$ and the source terms $F$, $H$  lie
in the spaces:
\gath{tplinIVPc}{
\psi^{\mu},\, b^{\mu\nu}  \in \Xc_T^{s+1}(\Tbb^n),  \quad D\psi^{\mu},\, D b^{\mu\nu} \in \Xc_T^{s}(\Tbb^n),
\quad D^2 \psi^{\mu} \in \Xc_T^{s-1}(\Tbb^n)
\intertext{and}
F\in \Xc_T^{s}(\Tbb^n), \quad H\in X_T^s(\Omega_1),
}
respectively,
\item $\phi$ is a smooth non-negative function on $\Tbb^n$ satisfying $\phi|_{B_{\rho}(x_{\pm})} =1$ and $\phi_{\Tbb^n\setminus(B_{2\rho}(x_+)\cup B_{2\rho}(x_-))}
=0$, where $x_+\in \Omega_1$ and $x_-\in \Omega_1^c$ are fixed points and $\rho>0$ is
any number satisfying $B_{3\rho}(x_+)\subset \Omega_1$ and $B_{3\rho}(x_{-})\subset \Omega_1^c$,
\item and the initial data
\eqn{tpidata}{
(\Ut_0,\Ut_1)\in \Hc^{2,s+1}(\Tbb^n)\times \Hc^{2,s}(\Tbb^n)
}
satisfies the compatibility conditions
\eqn{tplincompat}{
\Ut_\ell := \del{t}^\ell U |_{t=0} \in \Hc^{m_{s+1-\ell},s+1-\ell}(\Tbb^n), \qquad \ell=2,\ldots,s,
}
where the formal time derivatives are generated by the evolution equation that
results from removing the smoothing operators $S_\lambda$ from \eqref{tplinIVP.1}.
\end{enumerate}

As is clear from the proof of Theorem 3.2 of \cite{AnderssonOliynyk:2014}, the key step is to derive
$\lambda$-independent energy estimates for a one-parameter family of solutions $U^\lambda$.
From this one-parameter family, a solution to the IVP \eqref{tplinIVP.1}-\eqref{tplinIVP.2}
is then obtained by sending $\lambda \searrow 0$
and extracting a weakly convergence subsequence that converges to a solution. Once this
step is completed, the rest of the proof follows from the same arguments used to
prove Theorems 3.2 and 3.4 of \cite{AnderssonOliynyk:2014}. We omit these details.

To proceed, we set
\eqn{tlinIVP5}{
R= \norm{\del{}\psi}_{\Xc_T^{s}(\Tbb)}+ \norm{D\del{}\psi}_{\Xc_T^{s-1}(\Tbb)} + \norm{b}_{\Xc_T^{s}(\Tbb)}+ \norm{Db}_{\Xc_T^{s-1}(\Tbb)}.
}
Choosing $\ep >0$ sufficiently small\footnote{Small enough so that $J$ and $a^{\mu\nu}_\lambda$ satisfy \eqref{tlinIVPb.1}-\eqref{tlinIVPb.2}
 uniformly for
$\lambda \in (0,1]$.} and applying a standard existence theorem for linear hyperbolic equations,
for example see see \cite[Ch. 16, Proposition 1.7]{TaylorIII:1996}, we obtain a 1-parameter family of (unique) solutions
\eqn{tlinIVP6}{
U^\lambda \in \bigcap_{\ell=0}^{s+100} C^\ell([0,T),H^{s+100-\ell}(\Tbb^n)),\quad 0 < \lambda \leq 1,
}
to \eqref{tplinIVP.1}-\eqref{tplinIVP.2}.
Following \cite{AnderssonOliynyk:2014}, we derive bounds on $U^\lambda$ by first using elliptic estimates to estimate the
first $s-1$ time derivatives of $U^\lambda$, and then estimating the remaining $s$ and $s+1$ time derivatives using a hyperbolic
estimate.

\bigskip

\noindent \emph{Elliptic estimates:} We begin the derivation of the elliptic estimates by defining
\eqn{tlinIVP7}{
B^{\mu\nu} = \frac{A^{\mu\nu}- m^{\mu\nu}}{\ep}.
}
From the analyticity of the matrix inversion map $\text{Inv}(M) = M^{-1}$, it is not difficult to verify that the map
\eqn{tlinIVP8}{
(S_\lambda(b^{\mu\nu}),\del{\nu} S_\lambda(\psi^\mu),\ep) \longmapsto \frac{A^{\mu\nu}- m^{\mu\nu}}{\ep}
}
is analytic in a neighborhood of $(0,0,0)$. From this, the assumption $s>n/2+1$ and Proposition \ref{fpropB}, we then conclude that
\leqn{tlinIVP9}{
\norm{B}_{\Xc_T^s} \leq C(R)
}
for $\ep > 0$ small enough.
Differentiating $B^{\mu\nu}$, we also see that
\leqn{tlinIVP10}{
DB^{\mu\nu} = L^{\mu\nu}\cdot (S_\lambda(Db^{\mu\nu}),\del{\nu} S_\lambda(D\psi^\mu)),
}
where $L^{mu\nu}$ is a linear map that depends analytically on $(S_\lambda(b^{\mu\nu}),\del{\nu} S_\lambda(\psi^\mu),\ep)$.
Since $s-1>n/2$, we conclude directly from Proposition \ref{fpropB} that\footnote{It is worth noting that a more careful estimate,
obtained by differentiating \eqref{tlinIVP10} repeatedly in time followed by applying Theorem \ref{calcpropB}.(ii)
to the product terms and Proposition \ref{fpropB} to the terms $\del{t}^r L^{\mu\nu}$, shows that estimate \eqref{tlinIVP10} continues
to hold for $s>n/2$.}
\leqn{tlinIVP11}{
\norm{DB}_{\Xc_T^{s-1}} \leq C(R).
}

Due to the bounds \eqref{tlinIVP10} and \eqref{tlinIVP11} and the form of the evolution equations \eqref{tplinIVP.1}-\eqref{tplinIVP.2},
we can apply directly the elliptic estimates derived in the proof of Theorem 3.2 from \cite{AnderssonOliynyk:2014} (see in particular, equation (3.46)
from that article)
to conclude that $U^\lambda$ satisfies an estimate of the form
\alin{tlinIVP12}{
\norm{U^\lambda(t)}_{E^{s+1,s-1}(\Tbb^n)}& \leq \frac{c_L}{1-\ep c_L C(R)}\Bigl(\ep C(R)\bigl(\norm{U^\lambda(t)}_{E^{s+1,s-1}(\Tbb^n)} \\
&\hspace{1.cm} + \norm{U^\lambda_s(t)}_{E(\Tbb^n)} \bigr)    + \norm{F(t)}_{\Ec^{s-1}(\Tbb^n)}+ \norm{H(t)}_{E^{s-1}(\Omega_1)}
\Bigr),\quad 0\leq t < T,
}
for some constant $c_L$ independent of $\ep$ and $\lambda$, where here and below, we employ the notation
\eqn{tlinIVP13}{
U_\ell^\lambda := \del{t}^\ell U, \quad \ell \in \Zbb_{\geq 0}.
}
Choosing $\ep>0$ small enough, the above estimate yields
\leqn{tlinIVP14}{
\norm{U^\lambda(t)}_{E^{s+1,s-1}(\Tbb^n)} \leq 2 c_L\bigl(\norm{U^\lambda_s(t)}_{E(\Tbb^n)} + \norm{F(t)}_{\Ec^{s-1}(\Tbb^n)}+ \norm{H(t)}_{E^{s-1}(\Omega_1)}
\bigr),\quad 0\leq t < T.
}

\bigskip

\noindent \emph{Hyperbolic estimates:} Up to this point, we have taken over, essentially unchanged, our arguments
from \cite{AnderssonOliynyk:2014}. As in \cite{AnderssonOliynyk:2014}, we again use a hyperbolic
estimate to estimate the top time derivative. However, unlike \cite{AnderssonOliynyk:2014}, we cannot
use wave type estimate for $\del{t}^s U^\lambda$ as this would require estimating $s+1$ time
derivatives of $J$, which is one too many our purposes. Instead, we proceed by introducing a new
variable
\leqn{tlinIVP15}{
u_\beta = \Jch^\alpha_\beta \del{\alpha} U^\lambda
}
in order to recast the wave equation \eqref{tplinIVP.1} in first order form. Using the well known identity
\eqn{tlinIVP16}{
\del{\mu}(\det(J) \Jch^\mu_\nu) = 0
}
satisfied by Jacobian matrices, a short calculation shows that $u_\beta$ satisfies the symmetric
hyperbolic equation
\leqn{tlinIVP17}{
A^{\alpha\beta\mu}\del{\mu} u_\beta = \delta^\alpha_0\bigl[\det(J)\Jch^\mu_\gamma\del{\mu}
a^{\beta\gamma} u_\beta -\phi U^\lambda-S_\lambda(F)-S_\lambda(\chi_{\Omega_1}H)\bigr],
}
where
\eqn{tlinIVP18}{
A^{\alpha\beta\mu} = \det(J)\bigl(-\delta^\alpha_0 a^{\beta\gamma} -
\delta^\beta_0 a^{\alpha\gamma} + \delta^\gamma_0 a^{\alpha\beta}\bigr)\Jch^\mu_\gamma.
}

Differentiating \eqref{tlinIVP17} $s$-times with respect to $t$, we see that
\eqn{tlinIVP19}{
u^s_\beta = \del{t}^s u_\beta
}
satisfies the equation
\leqn{tlinIVP20}{
A^{\alpha\beta\mu}\del{\mu} u_\beta^s = \delta^\alpha_0 Y^\beta_0 u_\beta^s
+\sum_{\ell=0}^{s-1} \binom{s}{\ell}\bigl[-A^{\alpha\beta\mu}_{s-\ell} \del{\mu} u_\beta^\ell
+ \delta^\alpha_0 Y^\beta_{s-\ell}u^\ell_\beta\bigr]-\delta^\alpha_0\bigl[\del{t}^s(\phi U^\lambda)+S_\lambda(\del{t}^s F)+
S_\lambda(\chi_{\Omega_1} \del{t}^s H)\bigr],
}
where
\gath{tlinIVP21}{
u_\beta^\ell = \del{t}^\ell u_\beta, \qquad A^{\alpha\beta\gamma}_\ell = \del{t}^\ell A^{\alpha\beta\gamma},
\intertext{and}
Y^\beta_\ell = \del{t}^\ell\bigl(\det(J) \Jch^\mu_\gamma \del{\mu} a^{\beta\gamma}\bigr)
=\sum_{r=0}^\ell \binom{\ell}{r} \del{t}^{\ell-r}\bigl(\det(J) \Jch^\mu_\gamma)\del{t}^r \del{\mu}a^{\beta\gamma}.
}
Since \eqref{tlinIVP20} is symmetric hyperbolic, energy estimates, see
\cite[Ch. 16, \S 1]{TaylorIII:1996}, imply
that
\leqn{tlinIVP22}{
\del{t} \norm{u^s(t)}_{L^2(\Tbb^n)} \lesssim \bigl(1 + \norm{\del{\gamma}A^{\alpha\beta\gamma}(t)}_{L^\infty(\Tbb^n)}
+\norm{A^{\alpha\beta i}(t)}_{L^\infty(\Tbb^n)}\bigr)\norm{u^s(t)}_{L^2(\Tbb^n)} + \norm{\Fc(t)}_{L^2(\Tbb^n)},
}
where
\eqn{tlinIVP23}{
\Fc^\alpha =  \delta^\alpha_0 Y^\beta_0 u_\beta^s
+\sum_{\ell=0}^{s-1} \binom{s}{\ell}\bigl[-A^{\alpha\beta\mu}_{s-\ell} \del{\mu} u_\beta^\ell
+ \delta^\alpha_0 Y^\beta_{s-\ell}u^\ell_\beta\bigr]-\delta^\alpha_0\bigl[\del{t}^s(\phi U^\lambda) + S_\lambda(\del{t}^s F)+
S_\lambda(\chi_{\Omega_1} \del{t}^s H)\bigr].
}

In order to proceed from the energy estimate \eqref{tlinIVP22} to an effective bound on $\del{t}^s U^\lambda$, we
must first estimate the terms on the right hand side of \eqref{tlinIVP22}. We begin this
process by noting
that the estimate
\leqn{tlinIVP24}{
\norm{J}_{L^\infty(\Tbb^n)} + \norm{\Jch}_{L^\infty(\Tbb^n)} + \norm{\del{t}^\ell \Jch}_{\Hc^{0,s-\ell}(\Tbb^n)} \leq C(R), \quad 0\leq \ell \leq s,
}
follows directly from the analyticity of $\text{Inv}$ in
the neighborhood of the identity, Sobolev's inequality\footnote{$\norm{u}_{L^\infty(\Tbb^n)}\leq \max\{\norm{u}_{L^\infty(\Omega_1)},\norm{u}_{L^\infty(\Omega_1^c)}\}
\lesssim \max\{\norm{u}_{H^k(\Omega_1)},\norm{u}_{H^k(\Omega_1^c)}\} \lesssim \norm{u}_{\Hc^{0,k}}$ for $k>n/2$.},
and Propositions \ref{mollprop} and \ref{fpropB}.
Next, we differentiate $\Jch$ to get
\eqn{tlinIVP8}{
\del{\mu}\Jch = D\text{Inv}(J)\cdot \del{\mu}J.
}
Observing that $\del{\mu}J^\alpha_\beta = \ep\del{\mu}\del{\alpha} J_\lambda \psi^\beta$,
we obtain, using similar reasoning, the related estimates
\leqn{tlinIVP25}{
\norm{\del{t}\Jch}_{L^\infty(\Tbb^n)} + \norm{\del{t}^\ell \del{t}\Jch}_{\Hc^{0,s-1-\ell}(\Tbb^n)} +
\norm{D\Jch}_{L^\infty(\Tbb^n)}+\norm{\del{t}^\ell D\Jch}_{\Hc^{0,s-1-\ell}(\Tbb^n)} \leq C(R)\ep, \quad 0\leq \ell \leq s-1.
}
From \eqref{tlinIVP24} and \eqref{tlinIVP25} and Sobolev's inequality, we find that the estimate
\leqn{tlinIVP26}{
\bigl(1 + \norm{\del{\gamma}A^{\alpha\beta\gamma}(t)}_{L^\infty(\Tbb^n)}
+\norm{A^{\alpha\beta i}(t)}_{L^\infty(\Tbb^n)}\bigr) \leq C(R), \quad 0\leq t < T,
}
holds.

Continuing on, it is clear that the estimate
\eqn{tlinIVP27}{
\norm{Y_{s-\ell}(t)}_{\Hc^{0,\ell}(\Tbb^n)} \lesssim \norm{\det(J(t))\Jch(t)}_{\Ec^{s}(\Tbb^n)}
\norm{\del{}a(t)}_{\Ec^{s}(\Tbb^n)} \leq \ep C(R)\norm{\del{}b(t)}_{\Ec^s(\Tbb^n)}, \quad 0\leq \ell \leq s, \; 0\leq t < T,
}
follows directly from \eqref{tlinIVP25} and Theorem \ref{calcpropB}.(ii). This estimate followed
by another application of
Theorem \ref{calcpropB}.(ii) yields
\leqn{tlinIVP28}{
\norm{Y^\beta_{s-\ell}(t)u^\ell_\beta(t)}_{L^2(\Tbb^n)}
\lesssim \norm{Y_{s-\ell}(t)}_{\Hc^{0,\ell}(\Tbb^n)} \norm{u^\ell(t)}_{\Hc^{0,s-\ell}(\Tbb^n)} \leq \ep C(R)
\norm{\del{}b(t)}_{\Ec^s(\Tbb^n)}
\norm{u(t)}_{\Ec^s(\Tbb^n)}
}
for $0\leq \ell \leq s$ and $0\leq t < T$, while similar arguments show that
\leqn{tlinIVP29}{
\norm{A^{\alpha\beta\gamma}_{s-\ell}(t)\del{\gamma}u^\ell_\beta(t)}_{L^2(\Tbb^n)}
\leq C(R)\bigl(\norm{u(t)}_{\Ec^s(\Tbb^n)}+\norm{Du(t)}_{\Ec^{s-1}(\Tbb^n)}\bigr), \quad 0\leq \ell \leq s-1, \;\; 0\leq t < T.
}

Taken together, the estimates \eqref{tlinIVP22}, \eqref{tlinIVP26}, \eqref{tlinIVP28} and \eqref{tlinIVP29}
show, with the help of Proposition \ref{mollprop}, that $u^s_\beta$ satisfies the estimate
\eqn{tlinIVP30}{
\del{t} \norm{u^s(t)}_{L^2(\Tbb^n)} \lesssim C(R)\bigl(\beta(t)\norm{u(t)}_{\Ec^s(\Tbb^n)}+\norm{Du(t)}_{\Ec^{s-1}(\Tbb^n)}
+ \norm{F(t)}_{\Ec^s(\Tbb^n)} + \norm{H(t)}_{E^s(\Omega_1)} \bigr), \quad 0\leq t < T,
}
where
\eqn{betadef}{
\beta(t) = \bigl(1+\norm{\del{}b(t)}_{\Ec^s(\Tbb^n)}\bigr).
}
Integrating the above inequality in time, we find that
\lalin{tlinIVP31a}{
\norm{u^s(t)}_{L^2(\Tbb^n)}  \leq &\norm{u^s(0)}_{L^2(\Tbb^n)} + C(R)\int_0^t \beta(\tau) \norm{u(\tau)}_{\Ec^s(\Tbb^n)}
\notag \\
&\qquad + \norm{Du(\tau)}_{\Ec^{s-1}(\Tbb^n)}
+ \norm{F(\tau)}_{\Ec^s(\Tbb^n)} + \norm{H(\tau)}_{E^s(\Omega_1)}\, d\tau,\quad 0\leq t < T. \label{tlinIVP31}
}

Next, differentiating \eqref{tlinIVP15} gives
\leqn{tlinIVP32}{
u^\ell_\beta = \sum_{r=0}^\ell \binom{\ell}{r}\bigl[\del{t}^{\ell-r}\Jch^0_\beta
U^\lambda_{r+1}+ \del{t}^{\ell-r} \Jch^i_\beta \del{i} U^\lambda_r\bigr].
}
Using \eqref{tlinIVP24} and Theorem \ref{calcpropB}.(ii), we estimate \eqref{tlinIVP32} by
\eqn{tlinIVP33}{
\norm{u^\ell(t)}_{\Hc^{0,s-\ell}(\Tbb^n)} \leq C(R)\sum_{a=0}^1\norm{D^a U^\lambda(t)}_{\Ec^{0,s+1-a}(\Tbb^n)},
\quad 0\leq t < T, \quad  0\leq \ell \leq s.
}
Differentiating \eqref{tlinIVP32} and arguing in a similar fashion, we also get
\eqn{tlinIVP34}{
\norm{Du^\ell(t)}_{\Hc^{0,s-1-\ell}(\Tbb^n)} \leq  C(R)\sum_{a=0}^2\norm{D^a U^\lambda(t)}_{\Ec^{0,s+1-a}(\Tbb^n)},
\quad 0\leq t < T, \quad 0 \leq \ell \leq s-1.
}
Combining these two estimates yields
\leqn{tlinIVP35}{
\norm{u(t)}_{\Ec^s(\Tbb^n)} + \norm{Du(t)}_{\Ec^{s-1}(\Tbb^n)} \leq C(R)
\norm{U^\lambda(t)}_{E^{s+1}(\Tbb^n)}, \quad 0\leq t <T.
}

Setting $\ell=s$ in \eqref{tlinIVP32}, we find, after solving for $U^\lambda_{s+1}$ and
$\del{i} U^\lambda_s$, that
\lalin{tlinIVP36}{
\del{t}U^\lambda_s &= -s J^\beta_0 \del{t}\Jch^0_\beta U^\lambda_s +
J^\beta_0 u^s_\beta + \Rc_0 \label{tlinIVP36.1}
\intertext{and}
\del{i} U^\lambda_s &= -s J^\beta_i \del{t} \Jch^0_\beta U^\lambda_s
+ J^\beta_i u^s_\beta + \Rc_i, \label{tlinIVP36.2}
}
where
\eqn{tlinIVP37}{
\Rc_\omega = -s J^\beta_\omega \del{t}\Jch^i_\beta \del{i} U^\lambda_{s-1}
+ \sum_{r=0}^{s-2} \binom{s}{r}\bigl[J^\beta_\omega \del{t}^{s-r}\Jch^0_\beta U^\lambda_{r+1}
+ J_\omega^\beta \del{t}^{s-r} \Jch^{i}_\beta \del{i} U^\lambda_r \bigr].
}

Continuing on, we estimate $\Rc_\omega$ by
\leqn{tlinIVP38}{
\norm{\Rc(t)}_{L^2(\Tbb^n)} \leq \ep C(R)\bigl(\norm{U^\lambda(t)}_{\Ec^{s+1,s-2}(\Tbb^n)} + \norm{D U^\lambda(t)}_{\Ec^{s,s-1}(\Tbb^n)}\bigr) \leq \ep C(R)
\norm{U^\lambda(t)}_{E^{s+1,s-1}(\Tbb^n)}
}
for $0\leq t < T$
using \eqref{tlinIVP24},\eqref{tlinIVP26} and
Theorem \ref{calcpropB}.(ii). Writing \eqref{tlinIVP36.1} as
\eqn{tlinIVP39}{
\del{t}\bigl(e^\omega U^\lambda_s\bigr) = e^{\omega}\bigl[J^\beta_0 u^s_\beta + \Rc_0\bigr]
\quad \text{with} \quad
\omega(t) = s\int_0^s J^\beta_0(\tau)\del{t} \Jch^0_\beta(\tau)\, d\tau,
}
allows us, after integrating in time, to express $U^\lambda_s$ as
\eqn{tlinIVP40}{
U^\lambda_s(t) = U^\lambda_s(0) + e^{-\omega(t)}\int_0^t
e^{\omega(\tau)}\bigl[J^\beta_0(\tau) u^s_\beta(\tau)+ \Rc_0(\tau)\bigr] \, d\tau,
}
which we can, with the help of \eqref{tlinIVP24}, \eqref{tlinIVP35} and \eqref{tlinIVP38}, estimate by
\eqn{tlinIVP41}{
\norm{U^\lambda_s(t)}_{L^2(\Tbb^n)} \leq \norm{U^\lambda_s(0)}_{L^2(\Tbb^n)}
+ C(R)\int_0^t \norm{U^\lambda(\tau)}_{E^{s+1}(\Tbb^n)}\, d\tau, \quad 0\leq t < T.
}
Using this estimate together with \eqref{tlinIVP24}, \eqref{tlinIVP25} and
\eqref{tlinIVP38}, we see from \eqref{tlinIVP36.1} and \eqref{tlinIVP36.2} that
\lalin{tlinIVP42a}{
\norm{U^\lambda_s(t)}_{E(\Tbb^n)} \leq C(R)\biggl(&\norm{U^\lambda_s(0)}_{L^2(\Tbb^n)}
+ \ep \norm{U^\lambda(t)}_{E^{s+1,s-1}(\Tbb^n)}\notag \\
&\hspace{2.0cm} + \norm{u^s(t)}_{L^2(\Tbb^n)} +
\int_0^t \norm{U^\lambda(\tau)}_{E^{s+1}(\Tbb^n)}\, d\tau \biggr), \quad 0\leq t < T. 
\notag
}
From this estimate and \eqref{tlinIVP14}, it follows, after choosing $\ep>0$ small enough, that
\alin{tlinIVP42}{
\norm{U^\lambda_s(t)}_{E(\Tbb^n)} \leq C(R)&\biggl(\norm{U^\lambda_s(0)}_{L^2(\Tbb^n)}
+ \norm{F(t)}_{\Ec^{s-1}(\Tbb^n)} \\ &+ \norm{H(t)}_{E^{s-1}(\Omega_1)}
 + \norm{u^s(t)}_{L^2(\Tbb^n)} +
\int_0^t \norm{U^\lambda(\tau)}_{E^{s+1}(\Tbb^n)}\, d\tau \biggr),
}
which in turn, implies, again with the help of \eqref{tlinIVP14}, that
\lalin{tlinIVP43}{
\norm{U^\lambda(t)}_{E^{s+1}(\Tbb^n)} \leq C(R)&\biggl(\norm{U^\lambda_s(0)}_{L^2(\Tbb^n)}
+ \norm{F(t)}_{\Ec^{s-1}(\Tbb^n)} \notag \\
&+ \norm{H(t)}_{E^{s-1}(\Omega_1)} + \norm{u^s(t)}_{L^2(\Tbb^n)} +
\int_0^t \norm{U^\lambda(\tau)}_{E^{s+1}(\Tbb^n)}\, d\tau \biggr) \label{tlinIVP43.1}
}
for $0\leq t < T$. Writing
\eqn{tlinIVP44}{
F(t) = F(0) + \int_0^t \del{t}F(\tau)d\tau \AND H(t) = H(0) + \int_0^t \del{t}H(\tau)d\tau,
}
it then follows directly from \eqref{tlinIVP43.1} that
\alin{tlinIVP43}{
\norm{U^\lambda(t)}_{E^{s+1}(\Tbb^n)} \leq C(R)&\biggl(\norm{U^\lambda_s(0)}_{L^2(\Tbb^n)}
+ \norm{F(0)}_{\Ec^{s-1}(\Tbb^n)}
+ \norm{H(0)}_{E^{s-1}(\Omega_1)} \\
&+ \norm{u^s(t)}_{L^2(\Tbb^n)} +
\int_0^t \norm{U^\lambda(\tau)}_{E^{s+1}(\Tbb^n)} + \norm{F(\tau)}_{\Ec^{s}(\Tbb^n)} +
\norm{H(\tau)}_{E^s(\Omega_1)} \, d\tau \biggr)
}
for $0\leq t < T$. The above estimate in conjunction with \eqref{tlinIVP31} and \eqref{tlinIVP35} then
yields
\alin{tlinIVP43}{
\norm{U^\lambda(t)}_{E^{s+1}(\Tbb^n)} \leq C(R)&\biggl(\norm{U^\lambda(0)}_{E^{s+1}(\Tbb^n)}
+ \norm{F(0)}_{\Ec^{s-1}(\Tbb^n)}
+ \norm{H(0)}_{E^{s-1}(\Omega_1)} \\
& +
\int_0^t \beta(\tau)\norm{U^\lambda(\tau)}_{E^{s+1}(\Tbb^n)} + \norm{F(\tau)}_{\Ec^{s}(\Tbb^n)} +
\norm{H(\tau)}_{E^s(\Omega_1)} \, d\tau \biggr)
}
for $0\leq t < T$, independent of $\lambda \in (0,1]$. Applying Gronwall's inequality, we obtain
the key estimate

\lalin{tlinIVP44}{
\norm{U^\lambda(t)}_{E^{s+1}(\Tbb^n)}  \leq C(R,c_L)&e^{C(R,c_L)\int_0^t\beta(\tau)\,d\tau}
\biggl(\norm{U^\lambda(0)}_{E^{s+1}(\Tbb^n)}
+ \norm{F(0)}_{\Ec^{s-1}(\Tbb^n)} \notag \\
+ \norm{H(0)}_{E^{s-1}(\Omega_1)}&
 +
\int_0^t \beta(\tau)\norm{U^\lambda(\tau)}_{E^{s+1}(\Tbb^n)} + \norm{F(\tau)}_{\Ec^{s}(\Tbb^n)} +
\norm{H(\tau)}_{E^s(\Omega_1)} \, d\tau \biggr) . \label{tlinIVP44.1}
}

From this point, we can proceed exactly as in the proof of \cite[Theorem 3.2]{AnderssonOliynyk:2014}
and send $\lambda \searrow 0$, and obtain from the estimate \eqref{tlinIVP44.1} a unique solution $U\in CX^{s+1}_T(\Tbb^n)$
to the IVP \eqref{tplinIVP.1}-\eqref{tplinIVP.2}, with $\lambda =0$, that satisfies the energy estimate
\alin{tlinIVP45}{
\norm{U(t)}_{E^{s+1}(\Tbb^n)}  \leq C(R,c_L)&e^{C(R,c_L)\int_0^t\beta(\tau)\,d\tau}
\biggl(\norm{U^\lambda(0)}_{E^{s+1}(\Tbb^n)}
+ \norm{F(0)}_{\Ec^{s-1}(\Tbb^n)}
   \\
+ \norm{H(0)}_{E^{s-1}(\Omega_1)}&
 +
\int_0^t \beta(\tau)\norm{U^\lambda(\tau)}_{E^{s+1}(\Tbb^n)} + \norm{F(\tau)}_{\Ec^{s}(\Tbb^n)} +
\norm{H(\tau)}_{E^s(\Omega_1)} \, d\tau \biggr).
}
We can then argue as in the proof of \cite[Theorem 3.4]{AnderssonOliynyk:2014} to
obtain the existence of a unique solution $U\in CX^{s+1}_T(\Tbb^n)$ to the IVP consisting
of \eqref{tlinIVP.1}, \eqref{tlinIVP.2} and \eqref{tlincompat} that satisfies an energy estimate
of the form
\alin{tlinIVP45}{
\norm{U(t)}_{E^{s+1}(\Tbb^n)}  \leq C(\alpha,\mu)&
\biggl(\norm{U^\lambda(0)}_{E^{s+1}(\Tbb^n)}
+ \norm{F(0)}_{\Ec^{s-1}(\Tbb^n)}
   \\
+ \norm{H(0)}_{E^{s-1}(\Omega)}&
 +
\int_0^t \sigma(\tau)\norm{U^\lambda(\tau)}_{E^{s+1}(\Tbb^n)} + \norm{F(\tau)}_{\Ec^{s}(\Tbb^n)} +
\norm{H(\tau)}_{E^s(\Omega)} \, d\tau \biggr),
}
where
\gath{tlinIVP46}{
\alpha = \norm{\del{}\Psi}_{\Xc_T^{s}(\Tbb^n)}+ \norm{D\del{}\Psi}_{\Xc_T^{s-1}(\Tbb^n)} +
\norm{a}_{\Xc_T^{s}(\Rbb)}+ \norm{Da}_{\Xc_T^{s-1}(\Tbb^n)}, \\
\sigma(t) = \bigl(1+\norm{\del{}a(t)}_{\Ec^s(\Tbb^n)}\bigr), \qquad
\mu = \int_0^T \sigma(\tau)\, d\tau,
}
and $\gamma,\kappa$ are as defined previously in \eqref{tlinIVPb.1}-\eqref{tlinIVPb.2}.
\end{proof}

\subsect{divwave}{Wave equations in divergence form on $\Omega$}

The class of non-linear wave equations in divergence form on bounded domains that we be of interest to us are:
\lalin{koch1}{
\del{\alpha}\bigl(L^\alpha_\Ac(\xv,v,u,\del{} u)\bigr) &= w_\Ac(\xv,v,u,\del{}u) \hspace{0.3cm} \text{in $[0,T)\times \Omega$}, \label{koch1.1} \\
\nu_\alpha L^\alpha(\xv,v,u,\del{} v\bigr) & = 0 \hspace{2.35cm} \text{in $[0,T)\times \del{}\Omega$}, \label{koch1.2} \\
(u,\del{0}u) &= (u_0,u_1) \hspace{1.375cm} \text{in $\{0\}\times \Omega$}, \label{koch1.3}
}
where
\begin{enumerate}[(i)]
\item $\Omega$ is a bounded, open set in $\Tbb^n$ with $C^\infty$ boundary,
\item $\nu_\alpha = \delta_\alpha^i\nu_i$, where $\nu_i$ is the outward pointing unit conormal to $\del{}\Omega$,
\item the calligraphic indices, e.g. $\Ac$, $\Bc$, $\Cc$, run from $1$ to $N$,
\item $u=(u^\Ac(\xv))$ and $v=(v^\Ac(\xv))$ are $\Rbb^N$-valued maps,
\item the functions $L^\alpha_\Ac(\xv,v,u,\del{} u)$ and $w_\Ac(\xv,v,u,\del{}u)$ are smooth for
\eqn{koch2}{
(\xv,v,u,\del{}u) \in \Omega \times \Uct \times \Vct \times \Wct
}
for some open sets $\Uct \in \Rbb^N$, $\Vct \in \Rbb^N$ and $\Wct \in \Rbb^{(n+1)\times N}$,
\item the derivatives
\eqn{koch3}{
L^{\alpha\beta}_{\Ac\Bc}(\xv,v,u,\del{}u) :=
\frac{\del{}L^\beta_\Bc(\xv,v,u,\del{}u)}{\del{}(\del{\alpha} u^\Ac)\;}
}
satisfy the symmetry condition
\eqn{koch4}{
L^{\alpha\beta}_{\Ac\Bc} = L^{\beta\alpha}_{\Bc\Ac},
}
and there exists open sets
\eqn{koch5}{
\Uc \subset \overline{\Uc} \subset \Uct,\quad \Vc \subset \overline{\Vc} \subset \Vct \AND
\Wc \subset \overline{\Wc} \subset \Wct,
}
and a $\kappa_0>0$ such that
\eqn{koch6}{
\xi^i L^{00}_{ij}(\xv,v,u,\del{}u) \xi^j \leq  -\kappa_0 |\xi|^2
}
for all $(\xv,v,u,\del{}u,\xi)$
$\in$ $[0,T)$ $\times$ $\Omega$ $\times$ $\Uc$ $\times$ $\Vc\times$ $\Wc$ $\times$ $\Rbb^n$,
\item there exist constants $\kappa_1>0$, $\gamma \in \Rbb$ such that $L^{ij} = (L^{ij}_{\Ac\Bc})$ satisfies the
coercivity condition:
\eqn{koch7}{
\ip{\del{i}\zeta}{L^{ij}((t,\cdot),u(t),v(t),\del{}v(t))\del{j} \zeta}_{L^2(\Omega)}
\geq \kappa_1 \norm{\zeta}^2_{H^1(\Omega)} - \gamma \norm{\zeta}^2_{L^2(\Omega)}
}
for each $t\in [0,T)$, $\zeta\in C^1(\overline{\Omega})$, and $u\in C^0([0,T]\times \overline{\Omega})$ and $v\in C^1([0,T]\times \overline{\Omega})$ satisfying
\eqn{koch8}{
(u(t,x),v(t,x),\del{}v(t,x))\in \Uc\times\Vc\times \Wc, \quad \forall\; x\in \Omega,
}
\item  the initial data
\eqn{koch9}{
(u_0,u_1)\in H^{s+1}(\Omega)\times H^s(\Omega), \quad  s\in\Zbb{>n/2+1},
}
satisfies
\eqn{koch10}{
\bigl(u_0(x),(u_1(x),Du_0(x))\bigr) \in \Vc\times \Wc, \quad \forall\; x\in \Omega,
}
and the higher time derivatives $\del{t}^\ell u|_{t=0}$ generated from this initial data
through
formally differentiating the evolution equation \eqref{koch1.1}-\eqref{koch1.2}, satisfy the
\emph{compatibility conditions}
\eqn{koch11}{
\del{t}^\ell \!\bigl(\nu_\alpha L^\alpha_\Ac(\xv,v,u,\del{} u)\bigr)\bigr|_{t=0} \in H^{s-\ell}(\Omega)\cap H^1_0(\Omega),
\quad \ell=1,2,\ldots,s-1,
}
and
\item $v\in CY_T^{s+1}(\Omega)$ and $v(\xv)\in \Uc$ for all $\xv \in [0,T)\times \Omega$.
\end{enumerate}

The following local existence and uniqueness theorem for solutions to the IBVP \eqref{koch1.1}-\eqref{koch1.3}
follows directly from the arguments used in the proof of Theorem 1.1 from \cite{Koch:1993}.
\begin{thm}\label{kochthm}
Under the assumptions (i)-(ix) above, there exists a
\eqn{kochthm2}{
\delta_*=\delta_*\biggl(\norm{u(0)}_{E^{s+1}(\Omega)},\norm{v(0)}_{E^{s+1}(\Omega)}, \int_0^T\norm{v(\tau)}_{E^{s+1}(\Omega)}\,d\tau\biggr) \in \biggl(\frac{1}{T},\infty\biggr)
}
and a unique
solution $u\in CY^{s+1}_{T_*}(\Omega)$ to the IBVP \eqref{koch1.1}-\eqref{koch1.3}, where $T_*=\frac{1}{\delta_*}$. Moreover,
this solutions satisfies the energy estimate
\eqn{kochthm1}{
\norm{u(t)}_{E^{s+1}(\Omega)} \leq C\bigl(\kappa_0,\kappa_1,
\gamma,\norm{u(0)}_{E^{s+1}(\Omega)},\norm{v(0)}_{E^{s+1}(\Omega)},\mu(t)\bigr)
}
for $0\leq t < T_*$, where
\eqn{kochthm2}{
\mu(t) = \int_0^t \norm{u(\tau)}_{E^{s+1}(\Omega)}+
\norm{v(\tau)}_{E^{s+1}(\Omega)}\, d\tau.
}

\end{thm}

\sect{lemat}{Local existence proofs}

\subsect{rlemat}{Proof of Theorem \ref{redlocthm}}

\subsubsect{rlemat:exist}{Existence}
To prove existence of solutions to the reduced Einstein-Elastic system, we begin by letting
\alin{fderivs}{
\lambdah^\ell_{\mu\nu} := \del{0}^\ell \gt_{\mu\nu} \bigl|_{X^0=0}, \quad \ell=2,\ldots,s+1, \AND
\phih_\ell^i :=  \del{0}^\ell \phi^i_{\mu\nu} \bigl|_{X^0=0}, \quad \ell=2,\ldots,s
}
denote the higher order time derivatives generated from the initial data \eqref{matEinElasA.2} and \eqref{matEinElasA.5}
through formally differentiating the field equations \eqref{matEinElasA.1} and \eqref{matEinElasA.3} with respect to $X^0$ at $X^0=0$. By assumption, these satisfy the compatibility conditions \eqref{EEcompat.1}-\eqref{EEcompat.2}.

Next, we set
\alin{rEEBR}{
\Bc_{R,1}&:= \bigl\{\,\gt_{\mu\nu} \in CX^{s+1}_T(\Rbb^n)\: \bigl| \:
\norm{\gt}_{X_T^{s+1}(\Tbb^n)}\leq R,\quad \del{0}^\ell \gt_{\mu\nu} \bigl|_{X^0=0} = \lambdah^\ell_{\mu\nu}, \;\; \ell=0,\ldots,s+1\,
\bigr\},
\intertext{and}
\Bc_{R,2}&:= \bigl\{\,\phi^i \in CY^{s+1}_T(\Omega)\: \bigl| \:
\norm{\phiv}_{X_T^{s+1}(\Omega)}\leq R,\quad \del{0}^\ell  \phi^i_{\mu\nu} \bigl|_{X^0=0} = \phih^i_\ell, \;\; \ell=0,\ldots,s\,
\bigr\},
}
and we define a map
\leqn{rEEJmapA}{
J_T(h) = (J_{T,1}(h),J_{T,2}(h)):=(\gt,\phi)
}
that maps
\eqn{rEEJmapB}{
h = (h_{\mu\nu}) \in \Bc_{R,1} \subset CX^{s+1}_T(\Tbb^n)
}
to $(\gt,\phiv) =(\gt_{\mu\nu},\phi^i)$, where
\eqn{rEEJmapC}{
(\gt,\phiv) \in CX^{s+1}_T(\Rbb^n)\times CY^{s+1}_T(\Omega)
}
is the unique solution of the IBVP:
\lalin{rEEJmapD}{
\del{\Lambda}\bigl(A^{\Lambda\Gamma}\del{\Gamma} \gt_{\mu\nu}\bigr) &=
\det(\Jt(\phiv))Q_{\mu\nu}(h,\Jtch(\phiv)\del{} h) + \chi_\Omega\Tc_{\mu\nu}(\Xv,h,\phiv,\del{}\phiv)
 \text{\hspace{0.3cm} in $[0,T)\times \Tbb^n$,} \label{rEEJmapD.1} \\
 (\gt_{\mu\nu},\del{0}\gt_{\mu\nu}) &= (\lambdah^0_{\mu\nu},\lambdah^1_{\mu\nu})
 \text{\hspace{6.15cm} in $\{0\}\times \Tbb^n$,} \label{rEEJmapD.2}\\
\del{\Lambda}(L^\Lambda_i(\Xv,h,\phiv,\del{}\phiv)) &= w_i(\Xv,h,\phiv,\del{}\phiv) \text{\hspace{5.35cm} in $[0,T) \times \Omega$,}
\label{rEEJmapD.3}\\
\nu_{\Lambda}L^\Lambda_i(\Xv,h,\phiv, \del{}\phiv) &= 0 \text{\hspace{7.5cm} in $[0,T) \times \del{}\Omega$,}
\label{rEEJmapD.4}\\
(\phi^i,\del{0}\phi^i) &= (\phih^i_0,\phih^i_1)   \text{\hspace{6.50cm} in $\{0\}\times\Omega$,}
\label{rEEJmapD.5}
}
where
\eqn{rEEJmapDa}{
\Jt(\phiv) = (\del{\Lambda} \phit^\mu), \quad \Jtch(\phiv)=\Jt(\phiv)^{-1},
}
$\phit$ is given by \eqref{phitdefA}, and
\eqn{rEEJmapDb}{
A^{\Lambda\Gamma} = \det{\Jt(\phiv)}\Jtch^\Lambda_\alpha(\phiv)a^{\alpha\beta}(h)\Jtch^\Gamma_\beta(\phiv)
}
with
\eqn{reeJmaDc}{
a(h)=(a^{\alpha\beta}(h)) := \sqrt{|h|}h^{-1}.
}

To see that the map \eqref{rEEJmapA} is well defined, we first observe that it
follows directly from Theorem \ref{kochthm} and the IBVP \eqref{rEEJmapD.3}-\eqref{rEEJmapD.5} that there exists a $T_*>0$
such that map
$J_{T,2}$ is well-defined for all $T\in (0,T_*)$ and satisfies
\leqn{rEEJmapE}{
J_{T,2}(\Bc_{R,1})\subset \Bc_{C(TR),2}, \quad 0<T < T_*.
}
Fixing
\leqn{rEEJmapEa}{
h \in \Bc_{R,1}, \quad  T\in (0,T_*),
}
and setting
\leqn{rEEJmapEb}{
\phi = J_{T,2}(h),
}
we observe that the bounds
\lgath{rEEJmapG}{
\norm{\Jt(\phiv)}_{\Xc^{s}_T(\Tbb^n)}+\norm{D(\Jt(\phiv))}_{\Xc^{s-1}_T(\Tbb^n)}  \leq C(TR), \label{rEEJmapG.0}\\
\norm{\del{}a}_{\Xc_T^{s}(\Tbb^n)}+\norm{\det(\Jt(\phiv))Q(\mu,\Jtch(\phiv)\del{}h)}_{\Xc^s_T(\Rbb^n)}
+ \norm{\Tc_{\mu\nu}(\cdot,h,\phiv,\del{}\phiv)}_{Y^s_T(\Omega)} \leq C(R,TR), \label{rEEJmapG.1}
\intertext{and}
\bigl\|\det(\Jt(\phiv))Q(h,\Jtch(\phiv)\del{} h)\bigl|_{X^0=0}\bigr\|_{\Ec^{s-1}(\Rbb^n)}+
\bigl\|\Tc_{\mu\nu}((X^0,\cdot),h,\phiv,\del{}\phiv)\bigl|_{X^0=0}\bigr\|_{E^{s-1}(\Omega)} \lesssim 1 \label{rEEJmapG.2}
}
follow directly from \eqref{EdefB}, \eqref{rEEJmapE}, \eqref{rEEJmapEa}, \eqref{rEEJmapEb} and Proposition \ref{fpropB}.
Writing $a(h)=(a^{\mu\nu}(h))$ and $h$ as
\eqn{rEEJmapH}{
a(h(t)) = a(h(0))+ \int_0^t D_h a(h(\tau))\cdot \del{t}h(\tau)\, d\tau \AND
h(t) = h(0)+\int_0^t \del{t} h(\tau)\,d\tau,
}
respectively, we see that
\eqn{rEEJmapH}{
\norm{h(t)}_{\Ec^{s}(\Tbb^n)}+  \norm{D h(t)}_{\Ec^{s-1}(\Tbb^n)} \lesssim 1 + TR,
}
and, with the help of Propositions \ref{fpropA} and \ref{fpropB}, that
\leqn{rEEJmapI}{
\norm{a(h)}_{\Xc^{s}_T(\Tbb^n)}+\norm{D(a(h))}_{\Xc^{s-1}_T(\Tbb^n)} \leq C(TR).
}
In view of the bounds \eqref{rEEJmapG.0}-\eqref{rEEJmapI}, we conclude via Theorem \ref{tlinthm} and Gronwall's inequality that
we can solve the IVP \eqref{rEEJmapD.1}-\eqref{rEEJmapD.2} on the time interval $[0,T_*)$, and that the solution
$\gt=(\gt_{\alpha\beta})$ satisfies the estimate
\eqn{rEEJampJ}{
\norm{\gt}_{X^{s+1}_T(\Tbb^n)}\leq c(TC(R)).
}
This shows that $J_{T,1}$ is well-defined for all  $T\in (0,T_*)$ and satisfies
\leqn{rEEJmapK}{
J_{T,1}(\Bc_{R,1}) \subset \Bc_{c(TC(R)),1}\quad 0<T\leq T_*.
}
Choosing $R > 0$ sufficiently large and setting $T=\min\{1/C(R),T_*\}$,
it is clear from  \eqref{rEEJmapE} and \eqref{rEEJmapK} that
\leqn{rEEJmapL}{
J_{T,1}(\Bc_{R,1})\subset \Bc_{R,1}
\AND
J_{T,2}(\Bc_{R,1})\subset \Bc_{R,2}.
}

Fixing $\gt^0\in \Bc_{R,1}$, we define two sequences $\{\gt^m\}_{m=1}^\infty$ and $\{\phiv_m\}_{m=1}^\infty$ by
\eqn{rEEJmapM}{
\gt^m = \uset{\text{$m$ times}}{\underbrace{J_{T,1}\circ \cdots \circ J_{T,1}}}(\gt^0) \AND \phiv_m=J_{T,2}(\gt^{m-1}).
}
By \eqref{rEEJmapL}, we see that
\leqn{rEEJmapM1}{
\{\gt^m\}_{m=1}^\infty \subset \Bc_{R,1} \AND \{\phiv_m\}_{m=1}^\infty\subset \Bc_{R,2},
}
while, from the definition of the maps $J_{R,1}$ and $J_{R,2}$, it is clear that the pair $(\gt^m,\phiv_m)$ solves
the IBVP:
\lalin{rEEJmapN}{
\del{\Lambda}\bigl(A^{\Lambda\Gamma}(\gt^m,\Jt(\phiv_{m-1}))\del{\Gamma} \gt^m_{\mu\nu}\bigr) &=
\det(\Jt(\phiv_{m-1}))Q_{\mu\nu}(\gt^{m-1},\Jtch(\phiv_{m-1})\del{} \gt^{m-1}) \notag \\
 & \hspace{1.9cm}+ \chi_\Omega\Tc_{\mu\nu}(\Xv,\gt^{m-1},\phiv_{m-1},\del{}\phiv_{m-1})
 \text{\hspace{0.3cm} in $[0,T)\times \Tbb^n$,} \label{rEEJmapN.1} \\
 (\gt^m_{\mu\nu},\del{0}\gt^m_{\mu\nu}) &= (\lambdah^0_{\mu\nu},\lambdah^1_{\mu\nu})
 \text{\hspace{5.525cm} in $\{0\}\times \Tbb^n$,} \label{rEEJmapN.2}\\
\del{\Lambda}(L^\Lambda_i(\Xv,\gt^{m-1},\phiv_m,\del{}\phiv_m)) &= w_i(\Xv,\gt^{m-1},\phiv_m,\del{}\phiv_m) \text{\hspace{3.6cm} in $[0,T) \times \Omega$,}
\label{rEEJmapN.3}\\
\nu_{\Lambda}L^\Lambda_i(\Xv,\gt^{m-1},\phiv_m, \del{}\phiv_m) &= 0 \text{\hspace{6.9cm} in $[0,T) \times \del{}\Omega$,}
\label{rEEJmapN.4}\\
(\phi^i_m,\del{0}\phi^i_m) &= (\phih^i_0,\phih^i_1)   \text{\hspace{5.925cm} in $\{0\}\times\Omega$.}
\label{rEEJmapN.5}
}

The bounds \eqref{rEEJmapM1} allows us to extract, via the sequential Banach-Alaoglu Theorem, weakly convergent subsequences of $\{\gt^m\}_{m=1}^\infty$ and
$\{\phiv_m\}_{m=1}^\infty$, which we again denote by $\{\gt^m\}_{m=1}^\infty$  and $\{\phiv_m\}_{m=1}^\infty$, that
converge weakly to
\lalin{wlimit}{
\gt^\infty &\in \bigcap_{\ell=0}^{s+1} W^{\ell,q}\big([0,T),\Hc^{m_{s+1-\ell},s+1-\ell}(\Tbb^n)\bigr) \label{wlimit.1}
\intertext{and}
\phiv_\infty &\in \bigcap_{\ell=0}^{s+1} W^{\ell,q}\big([0,T),H^{s+1-\ell}(\Omega)\bigr) \label{wlimit.2}
}
for any $q\in (1,\infty)$ as $m\rightarrow \infty$. Choosing $\ep>0$ small enough
so that $s-1-\ep>n/2$, we can using integral and fractional versions of the Rellich-Kondrachov Compactness Theorem
extract subsequences, again denoted
by $\{\gt^m\}_{m=1}^\infty$  and $\{\phiv_m\}_{m=1}^\infty$, such that
\lalin{rEEJmapO}{
\gt^m \longrightarrow \gt^\infty \quad &\text{in $L^q\bigl([0,T),H^{s-\ep}(\Omega)\cap H^{2-\ep}(\Tbb^n)\cap H^{s-\ep}(\Omega^c)\bigr)$,} \label{rEEJmapO.1}\\
\gt^m \longrightarrow \gt^\infty \quad &\text{in $W^{1,q}\bigl([0,T),H^{s-1-\ep}(\Omega)\cap H^{2-\ep}(\Tbb^n)\cap H^{s-1-\ep}(\Omega^c)\bigr)$,} \label{rEEJmapO.2}\\
\phiv_m \longrightarrow \phiv_\infty \quad &\text{in $L^q\bigl([0,T),H^{s-\ep}(\Omega)\bigr)$,} \label{rEEJmapO.3}
\intertext{and}
\phiv_m \longrightarrow \phiv_\infty \quad &\text{in $W^{1,q}\bigl([0,T),H^{s-1-\ep}(\Omega)\bigr)$,} \label{rEEJmapO.4}
}
as $m\rightarrow \infty$.

Testing \eqref{rEEJmapN.1} and \eqref{rEEJmapN.3} with $\zeta \in C^\infty_0\bigl([0,T),C^\infty(\Tbb^n,\Sbb{n+1})\bigr)$ and $\psi \in C^\infty_0\bigr([0,T), C^\infty(\overline{\Omega},\Rbb^n)\bigr)$, respectively, we find
that
\lalin{rrEEJmapP}{
&-\ip{\del{\Lambda}\zeta}{A^{\Lambda\Gamma}(\gt^m,\Jt(\phiv_{m-1}))\del{\Gamma} \gt^m}_{[0,T)\times\Tbb^n} =
\ip{\zeta}{\det(\Jt(\phiv_{m-1}))Q_{\mu\nu}(\gt^{m-1},\Jtch(\phiv_{m-1})\del{} \gt^{m-1})}_{[0,T)\times \Tbb^n} \notag \\
&\hspace{7.0cm} + \ip{\zeta}{\Tc_{\mu\nu}(\cdot,\gt^{m-1},\phiv_{m-1},\del{}\phiv_{m-1})}_{[0,T)\times\Omega}, \label{rrEJmapP.1}\\
&-\ip{\del{\Lambda}\psi}{L^\Lambda(\cdot,\gt^{m-1},\phiv_m,\del{}\phiv_m)}_{[0,T)\times\Omega} = \ip{\psi}{w(\cdot,\gt^{m-1},\phiv_m,\del{}\phiv_m)}_{[0,T)\times\Omega}.
\label{rrEJmapP.2}
}
Letting $m\rightarrow \infty$ in \eqref{rrEJmapP.1} and \eqref{rrEJmapP.2}, the strong convergence \eqref{rEEJmapO.1}-\eqref{rEEJmapO.4} in conjunction with the calculus inequalities from
Appendix \ref{calculus}
shows that the limits $(\gt^\infty,\phiv_\infty)$ satisfy
\lalin{rrEEJmapQ}{
&-\ip{\del{\Lambda}\zeta}{A^{\Lambda\Gamma}(\gt^{\infty},\Jt(\phiv_{{\infty}}))\del{\Gamma} \gt^{\infty}}_{[0,T)\times\Tbb^n} =
\ip{\zeta}{\det(\Jt(\phiv_{{\infty}}))Q_{\mu\nu}(\gt^{{\infty}},\Jtch(\phiv_{{\infty}})\del{} \gt^{{\infty}})}_{[0,T)\times \Tbb^n} \notag \\
&\hspace{7.0cm} + \ip{\zeta}{\Tc_{\mu\nu}(\cdot,\gt^{{\infty}},\phiv_{{\infty}},\del{}\phiv_{{\infty}})}_{[0,T)\times\Omega}, \\
&-\ip{\del{\Lambda}\psi}{L^\Lambda(\cdot,\gt^{{\infty}},\phiv_{\infty},\del{}\phiv_{\infty})}_{[0,T)\times\Omega} = \ip{\psi}{w(\cdot,\gt^{{\infty}},\phiv_{\infty},\del{}\phiv_{\infty})}_{[0,T)\times\Omega}.
}
Since $\zeta \in C^\infty_0\bigl([0,T),C^\infty(\Tbb^n,\Sbb{n+1})\bigr)$ and $\psi \in C^\infty_0\bigr([0,T), C^\infty(\overline{\Omega},\Rbb^n)\bigr)$ were chosen arbitrarily,
the limit $(\gt^\infty,\phiv_\infty)$  satisfies
\lalin{rEEJmapR}{
\del{\Lambda}\bigl(A^{\Lambda\Gamma}(\gt^\infty,\Jt(\phiv_\infty))\del{\Gamma} \gt^\infty_{\mu\nu}\bigr) &=
\det(\Jt(\phiv_\infty))Q_{\mu\nu}(\gt^\infty,\Jtch(\phiv_\infty)\del{} \gt^\infty) \notag \\
 & \hspace{2.9cm} + \chi_\Omega\Tc_{\mu\nu}(\Xv,\gt^\infty,\phiv_\infty,\del{}\phiv_\infty)
 \text{\hspace{0.3cm} in $[0,T)\times \Tbb^n$,} 
 \notag \\
\del{\Lambda}(L^\Lambda_i(\Xv,\gt^\infty,\phiv_\infty,\del{}\phiv_\infty)) &= w_i(\Xv,\gt^\infty,\phiv_\infty,\del{}\phiv_\infty) \text{\hspace{3.85cm} in $[0,T) \times \Omega$,}
\notag
\\
\nu_{\Lambda}L^\Lambda_i(\Xv,\gt^\infty,\phiv_\infty, \del{}\phiv_\infty) &= 0 \text{\hspace{6.875cm} in $[0,T) \times \del{}\Omega$,}
\notag \\
}
But by \eqref{rEEJmapM1}, $(\gt^\infty,\phiv_\infty)$ must also satisfy
\lalin{rEEJmapS}{
(\gt_{\mu\nu}^\infty,\del{0}\gt^\infty_{\mu\nu}) &= (\lambdah^0_{\mu\nu},\lambdah^1_{\mu\nu})
 \text{\hspace{0.35cm} in $\{0\}\times \Tbb^n$,} 
 \notag \\
(\phi^i_\infty,\del{0}\phi^i_\infty) &= (\phih^i_0,\phih^i_1)   \text{\hspace{0.725cm} in $\{0\}\times\Omega$,}
\notag
}
thereby establishing $(\gt^\infty,\phiv_\infty)$ a solution of the IVBP \eqref{matEinElasA.1}-\eqref{matEinElasA.5}.

The improved regularity of the solution $(\gt^\infty,\phiv_\infty)\in X^{s+1}_T(\Tbb^n)\times Y^{s+1}_T(\Omega)$ can
be established by repeatedly differentiating the field equations with respect to $X^0$ in order to express
the higher order time derivatives $\del{0}^\ell\gt_{\mu\nu}$ and
$\del{0}^\ell\phi^i$, $\ell = 2 \ldots s+1$, in terms of the lower time derivatives and spatial derivatives. The
bounds \eqref{wlimit.1}-\eqref{wlimit.2} together with the calculus inequalities from Appendix \ref{calculus}
can then be used to show that $(\gt^\infty,\phiv_\infty)\in X^{s+1}_T(\Tbb^n)\times Y^{s+1}_T(\Omega)$.

\subsubsect{rlemat:unique}{Uniqueness}
With existence established, we now turn to verifying the uniqueness of solutions to the IVBP \eqref{matEinElasA.1}-\eqref{matEinElasA.5}. Given
 two solutions $(\gt^a,\phiv_a)\in  X^{s+1}_T(\Tbb^n)\times Y^{s+1}_T(\Omega)$, $a=1,2$, we define
\eqn{uniqueA}{
\dot{\phiv}_a = \del{0}\phiv_a.
}
Differentiating the elastic field and boundary equations \eqref{matEinElasA.3} and \eqref{matEinElasA.4} with respect to $X^0$ shows
that the pair $(\phiv_a,\dot{\phiv}_a)$, $a=1,2$, satisfies
\lalin{uniqueB}{
\del{\Lambda}\bigl(L^{\Lambda\Gamma}_{aij}\del{\Gamma}\dot{\phi}^j_a + M^\Lambda_{ai}\bigr) &= W_{ai} \hspace{0.3cm}\text{in $[0,T)\times \Omega$,} \label{uniqueB.1}\\
\del{0}\phi^i_a &= \dot{\phi}^i_a \hspace{0.55cm}\text{in $[0,T)\times \Omega$,} \label{uniqueB.2}\\
\nu_{\Lambda}\bigl(L^{\Lambda\Gamma}_{aij}\del{\Gamma}\dot{\phi}^j_a + M^\Lambda_{ai}\bigr)&= 0 \hspace{0.75cm}\text{in $[0,T)\times \del{}\Omega$,} \label{uniqueB.3}
}
where
\alin{uniqueC}{
\xi_a &:= \bigl(\Xv,\gt_a,\phiv_a,(\dot{\phiv}_a,D\phiv_a)\bigr), \\
L^{\Lambda\Gamma}_{aij} &:= L^{\Lambda\Gamma}_{ij}(\xi_a),\\
M^\Lambda_{ai} &:= \frac{\del{}L^{\Lambda\Gamma}_{ij}}{\del{}X^0}(\xi_a)+
\frac{\del{}L^{\Lambda\Gamma}_{ij}}{\del{}\gt_{\mu\nu}}(\xi_a)\del{0}\gt^a_{\mu\nu}+ \frac{\del{}L^{\Lambda\Gamma}_{ij}}{\del{}\phi^j_a}(\xi_a)\dot{\phi}^j_a,
\intertext{and}
W_{ai} &:= \frac{\del{}w_i}{\del{}X^0}(\xi_a)+
\frac{\del{}w}{\del{}\gt_{\mu\nu}}(\xi_a)\del{0}\gt^a_{\mu\nu}+ \frac{\del{}w_i}{\del{}\phi^j_a}(\xi_a)\dot{\phi}^j_a+\frac{\del{}w_i}{\del{}\del{\Lambda}\phi^j_a}(\xi_a)\del{\Lambda}\dot{\phi}^j_a .
}
Following the usual approach to establishing uniqueness, we consider the differences
\eqn{uniqueH}{
(\delta \gt,\delta\phiv, \delta\dot{\phiv} ) := (\gt^2-\gt^1,\phiv_2-\phiv_1,\dot{\phiv}_2-\dot{\phiv}_1).
}
Since the $\gt^a_{\mu\nu}$, $a=1,2$, solve the reduced Einstein equations \eqref{matEinElasA.1}, a straight forward
straightforward calculation shows that $\delta \gt_{\mu\nu}$ satisfies
\leqn{uniqueI}{
\del{\Lambda}\bigl(A^{\Lambda\Gamma}_2 \del{\Gamma} \delta \gt_{\mu\nu}\bigr)
= \del{\Lambda}\bigl( [A^{\Lambda\Gamma}_1-A^{\Lambda\Gamma}_2]\del{\Gamma}\gt_{\mu\nu}^1\bigr)
+ \Qc^2_{\mu\nu}-\Qc^1_{\mu\nu} + \chi_{\Omega}\bigl(\Tc^2_{\mu\nu}-\Tc^1_{\mu\nu}\bigr)
\hspace{0.3cm}\text{in $[0,T)\times \Tbb^n$,}
}
where
\alin{uniqueJ}{
A^{\mu\nu}_a &:= A^{\mu\nu}\bigl(\Jt(\phiv_a),\gt_a\bigr),\\
\Qc_{\mu\nu}^a &:= \det\bigl(\Jt(\phiv_a)\bigr)Q_{\mu\nu}\bigl(\gt^a,\Jtch(\phiv_a)\gt^a\bigr)
\intertext{and}
\Tc_{\mu\nu}^a &:= \Tc_{\mu\nu}(\xi_a).
}
Defining
\eqn{uniqueJa}{
f^a_\beta = (f^a_{\beta\mu\nu}) := (\Jtch^\Lambda_\beta(\phiv_a)\del{\Lambda}\gt_{\mu\nu}^a),
}
a calculation similar to that used to derive \eqref{tlinIVP17} shows that the reduced Einstein equations \eqref{matEinElasA.1} imply the following evolution
equations for the $f^a_\beta$:
\leqn{uniqueJb}{
A^{\alpha\beta\Lambda}_a\del{\Lambda}f^a_{\beta\mu\nu} = F^{\alpha}_{a\mu\nu} + \chi_\Omega H^\alpha_{a\mu\nu} \hspace{0.3cm}\text{in $[0,T)\times \Tbb^n$,}
}
where
\alin{uniqueJe}{
A^{\alpha\beta\Lambda}_a &:= \det(\Jt(\phiv_a))\bigl(-\delta^\alpha_0 \sqrt{|\gt^a|}\gt_a^{\beta\gamma}-\delta^\beta_0\sqrt{|\gt^a|}\gt_a^{\alpha\gamma}
+ \delta^\gamma_0 \sqrt{|\gt^a|}\gt_a^{\alpha\beta} \bigr)\Jtch^\Lambda_\gamma(\phiv_a),\\
F^\alpha_{a\mu\nu} &:= \delta^\alpha_0\biggl[\det(\Jt(\phiv_a))\Jtch^\Lambda_\gamma \frac{\del{} \sqrt{|\gt|}\gt^{\beta\gamma}}{\del{}\gt_{\sigma\delta}}\Bigl|_{\gt=\gt^a}
\Jt_\Lambda^\omega f_{\omega \sigma\delta}^a
f^a_{\beta\mu\nu}- \det(\Jt(\phiv_a))Q_{\mu\nu}(\gt^a,f^a)\biggr],\\
H^\alpha_{a\mu\nu} &:= -\delta^\alpha_0 \Tc_{\mu\nu}(\xi_a),
}
and, as previously, we employ the notation $(\gt_a^{\alpha\beta}) = (\gt^a_{\alpha\beta})^{-1}$ and $|\gt^a|=-\det(\gt^a_{\mu\nu})$.
Setting
\eqn{uniqueJf}{
\delta f_\beta := f_\beta^2-f_\beta^1,
}
we see from \eqref{uniqueJb} that $\delta f_\beta$ satisfies
\leqn{uniqueJh}{
A^{\alpha\beta\Lambda}_2\del{\Lambda}\delta f_{\beta\mu\nu} = \bigl(A^{\alpha\beta\Lambda}_1-A^{\alpha\beta\Lambda}_2\bigr)\del{\Lambda}f^1_{\beta\mu\nu}+  F^{\alpha}_{2\mu\nu}-F^{\alpha}_{1\mu\nu} + \chi_\Omega( H^\alpha_{2\mu\nu}-H^\alpha_{1\mu\nu} )\hspace{0.3cm}\text{in $[0,T)\times \Tbb^n$.}
}

Following the same arguments used in the proof of Theorem \ref{tlinthm}, where we
view \eqref{uniqueI} as an elliptic equation for the purpose of estimating the $s-2$ time derivatives of $\delta \gt_{\mu\nu}$,
and use \eqref{uniqueJh} and hyperbolic estimates to estimate the top two time derivatives,
we obtain, with the help of
\leqn{uniqueL}{
\norm{\delta\gt(0)}_{E^{s+1}(\Tbb^n)} = \norm{\delta\phi(0)}_{E^{s+1}(\Tbb^n)}= \norm{\delta\dot{\phi}(0)}_{E^s(\Tbb^n)}=0,
}
the energy estimate
\leqn{uniqueM}{
\norm{\delta\gt(t)}_{E^s(\Tbb^n)} \lesssim \int_0^t \norm{\delta\gt(\tau)}_{E^s(\Tbb^n)}+ \norm{\delta\phi(\tau)}_{E^s(\Omega)}\, d\tau, \quad 0\leq t < T,
}
for $\delta\gt_{\mu\nu}$.

We also see from \eqref{uniqueB.1}-\eqref{uniqueB.3} that $\delta\phiv$ and  $\delta\dot{\phiv}$
satisfy
\lalin{uniqueN}{
\del{\Lambda}\bigl(L^{\Lambda\Gamma}_{2ij}\del{\Gamma}\delta\dot{\phi}^j + M^\Lambda_{2i}-M^\Lambda_{1i}+\bigl[L^{\Lambda\Gamma}_{2ij}-L^{\Lambda\Gamma}_{1ij}\bigr]
\del{\Gamma}\dot{\phi}_1^j\bigr) &= W_{2i}-W_{1i} \hspace{0.3cm}\text{in $[0,T)\times \Omega$,} \label{uniqueN.1}\\
\del{0}\delta\phi^i - \delta\dot{\phi}^i &=0 \hspace{1.75cm}\text{in $[0,T)\times \Omega$,} \label{uniqueN.2} \\
\nu_{\Lambda}\bigl(L^{\Lambda\Gamma}_{2ij}\del{\Gamma}\delta\dot{\phi}^j + M^\Lambda_{2i}-M^\Lambda_{1i}+\bigl[L^{\Lambda\Gamma}_{2ij}-L^{\Lambda\Gamma}_{1ij}\bigr]
\del{\Gamma}\dot{\phi}_1^j\bigr)&= 0 \hspace{1.75cm}\text{in $[0,T)\times \del{}\Omega$.} \label{uniqueN.3}
}
Applying the energy estimates from Theorem 2.4 of \cite{Koch:1993} to the this system,  we obtain, with
the help of the calculus inequalities from Appendix \ref{calculus} and \eqref{uniqueL},
the estimate
\leqn{uniqueO}{
\norm{\delta \phi(t)}_{E^{s-1}(\Omega)} + \norm{\delta\dot{\phi}(t)}_{E^{s-1}(\Omega)} \lesssim \int_0^t \norm{\delta\gt(\tau)}_{E^{s}(\Tbb^n)}+
\norm{\delta\dot{\phi}(\tau)}_{E^{s-1}(\Omega)}+ \norm{\delta\phi(\tau)}_{E^{s-1}(\Omega)}\, d\tau
}
for $0\leq t < T$. Viewing \eqref{uniqueN.1} and \eqref{uniqueN.3} as an elliptic equation for $\delta\dot{\phi}^i$,
and letting $\dsl\delta\dot{\phi}^i$ denote the derivatives tangential to the boundary $\del{}\Omega$, it follows,
after differentiating \eqref{uniqueN.1} and \eqref{uniqueN.3} tangentially, from
elliptic regularity, e.g. see \cite[Theorem A.4.]{Koch:1993}, and the calculus inequalities
from Appendix \ref{calculus} that
\lalin{uniqueP}{
\norm{\dsl\delta\dot{\phiv}(t)}_{H^{s-1}(\Omega)} \lesssim &\norm{\dsl\delta\dot{\phiv}(t)}_{H^{s-2}(\Omega)} + \norm{\del{0}\dot{\phiv}(t)}_{H^{s-2}(\Omega)} +
\norm{\nu(\dot{\phiv})(t)}_{H^{s-2}(\Omega)} \\
&+\norm{\delta\dot{\phiv}(t)}_{H^{s-2}(\Omega)} + \norm{\delta\phiv(t)}_{H^{s-1}(\Omega)}+\norm{\dsl \phiv(t)}_{H^{s-1}(\Omega)}+ \norm{\delta\gt(t)}_{E^{s}(\Tbb^n)},\label{uniqueP.1}
}
where $\nu = \nu^I\del{I}$ is the normal derivative. We also note that the
normal derivative can be estimated by treating \eqref{uniqueN.1} and \eqref{uniqueN.3} as an ODE for $\nu(\delta\dot{\phi}^i)$. Integrating in the normal direction followed by application of the calculus inequalities from Appendix \ref{calculus}
then yields the desired result.
From this type of normal estimate and the tangential estimate \eqref{uniqueP.1},
it is not difficult to verify that we obtain the estimate
\leqn{uniqueQ}{
\norm{D\delta\dot{\phiv}(t)}_{H^{s-1}(\Omega)} \lesssim \norm{\del{}\delta\dot{\phiv}(t)}_{H^{s-2}(\Omega)} +\norm{\delta\dot{\phiv}(t)}_{H^{s-2}(\Omega)} + \norm{\delta\phiv(t)}_{H^{s}(\Omega)}+ \norm{\delta\gt(t)}_{E^{s}(\Tbb^n)}
}
for the full spatial derivative of $\delta\dot{\phi}^i$. Integrating \eqref{uniqueN.2} in time, it follows from
\eqref{uniqueQ} that
\leqn{uniqueR}{
\norm{\delta\phiv(t)}_{H^{s}(\Omega)} \leq  \int_0^t \norm{\dot{\phiv}(\tau)}_{E^{s-1}(\Omega)} +  \norm{\delta\phiv(\tau)}_{H^{s}(\Omega)}+ \norm{\delta\gt(\tau)}_{E^{s}(\Tbb^n)} \, d \tau,
\quad 0\leq t < T.
}

Taken together, the inequalities \eqref{uniqueM}, \eqref{uniqueO}, and \eqref{uniqueR} imply, via Gronwall's inequality, that
\eqn{uniqueS}{
\norm{\delta \gt(t)}_{E^{s}(\Tbb)} + \norm{\delta\phi(t)}_{E^{s}(\Tbb)} = 0, \quad 0\leq t < T.
}
In particular, we have that $\gt^1_{\mu\nu}=\gt^2_{\mu\nu}$ on $[0,T)\times \Tbb^n$, and $\phit^i_1=\phit^i_2$ on $[0,T)\times \Omega$, and the uniqueness proof is complete.

\subsect{flemat}{Proof of Theorem \ref{fulllocthm}}

We begin the proof of Theorem \ref{fulllocthm} by fixing a solution $(\gt_{\mu\nu},\phi^i)\in X^{s+1}_T(\Sigma)\times Y^{s+1}_T(\Omega)$ of
the IBVP \eqref{matEinElasA.1}-\eqref{matEinElasA.5} from Theorem \ref{redlocthm}, and defining
the pull-back metric
\leqn{gbdef}{
\gb_{\Lambda\Gamma} := (\phit^*g)_{\Lambda\Gamma}  = \Jt_{\Lambda}^{\mu}\Jt^\nu_{\Gamma} \gt_{\mu\nu}.
}
We let $\Gammab^{\Lambda}_{\Gamma\Sigma}$, $\Gb_{\Lambda\Gamma}$, $\Rb_{\Lambda\Gamma}$,
and $\nablab_\Lambda$ denote the Christofell symbols, Einstein tensor, Ricci tensor,
and Levi-Civita
connection of the metric \eqref{gbdef}, respectively. Next, we set
\eqn{TtdefA}{
\Tt_{\mu\nu} = T_{\mu\nu}\circ \phi,
}
and recall, by assumption, that
\eqn{TtdefB}{
\Tt_{\mu\nu} = \Tt_{\mu\nu}\bigl(\Xv,\gt,\phiv,\del{}\phiv\bigr),
}
where $\Tt_{\mu\nu}$ is smooth for $\bigl(\Xv,\gt,\phiv,\del{}\phiv\bigr)$$\in$$[0,T)$$\times$$\Omega$$\times$$\Uct$$\times$$\Vct$$\times$$\Wct$.
We also let
\eqn{zetatdef}{
\zetat_{\mu} = \gt_{\mu\gamma}\gt^{\alpha\beta}\Gammat^\gamma_{\alpha\beta},
}
where
\eqn{Christt}{
\Gammat^\gamma_{\alpha\beta} := \Gamma^\gamma_{\alpha\beta}\circ\phi = \Half \gammat^{\gamma\lambda}\bigl(\Jtch_\alpha^\Lambda \del{\Lambda}\gt_{\gamma\beta}
+\Jtch^{\Lambda}_\beta \del{\Lambda}\gt_{\alpha\gamma} - \Jtch^\Lambda_\gamma \del{\Lambda} \gt_{\alpha\beta}\bigr),
}
and we define the pull-back of the stress-energy tensor $T_{\mu\nu}$ and
the co-vector field $\zeta_\mu$  by
\leqn{stressbzetab}{
\Tb_{\Lambda\Gamma} := (\phit^* T)_{\Lambda\Gamma} = \Jt_\Lambda^\mu \Jt_\Gamma^\nu \Tt_{\mu\nu} \AND
\zetab_{\Lambda} := (\phit^*\zeta)_{\Lambda} = \Jt_\Lambda^\mu \zetat_\mu,
}
respectively.
Since \eqref{matEinElasA.1}, \eqref{matEinElasA.3}, and \eqref{matEinElasA.4} are equivalent to
\eqref{redEinElasA.1}, \eqref{redEinElasA.2}, and \eqref{EinElasB}, it follows from
from the definitions \eqref{gbdef} and \eqref{stressbzetab}, and a simple calculation that
\lalin{EEbarA}{
\Gb_{\Lambda\Gamma} + \Half\bigl(-\nablab_{\Lambda}\zetab_\Gamma - \nablab_{\Gamma}\zetab_\Lambda +\nabla_{\Sigma}\zetab^\Sigma \gb_{\Lambda\Gamma} \bigr) &=
2\kappa \chi_\Omega \Tb_{\Lambda\Gamma} \text{\hspace{0.2cm} in $[0,T)\times \Tbb^n$}, \label{EEbarA.1}\\
\nablab^{\Lambda}\Tb_{\Lambda\Gamma} &= 0 \text{\hspace{1.5cm} in $[0,T)\times \Omega$}, \label{EEbarA.2}\\
\nb^\Lambda \Tb_{\Lambda\Gamma} &= 0 \text{\hspace{1.5cm} in $[0,T)\times \del{}\Omega$}, \label{EEbarB}
}
where $\nb_{\Lambda}$ $=$ $\frac{1}{\sqrt{\gb^{IJ}\nu_I\nu_J}}\delta_\Lambda^K \nu_K$.

To proceed, we define a smoothed version of the diffeomorphism $\phit$ by
\leqn{phitsmooth}{
\uset{\lambda}{\phit}^\mu(X^0,X) = \begin{cases} X^0 & \text{if $\mu=0$} \\
S_\lambda \Ebb_\Omega(\phi^i) & \text{if $\mu=i$}
\end{cases},
}
where $S_\lambda$ is the smoothing operator from Proposition \ref{mollprop}, and
$\lambda_0$ is chosen small enough so that \eqref{phitsmooth} continues to define a diffeomorphism for
$\lambda \in (0,\lambda_0]$.
We then set
\eqn{Jtsmooth}{
\uset{\lambda}{\Jt} = \bigl(\uset{\lambda}{\Jt}^\mu_\Lambda\bigr) := \bigl(\del{\Lambda}\uset{\lambda}{\phit}^\mu\bigr) \AND \uset{\lambda}{\Jtch} = \uset{\lambda}{\Jt}^{-1},
}
and define a smoothed version of the metric $\gb_{\Lambda\Gamma}$ and the covector field $\zetab_{\Lambda}$ by
\eqn{gbsmoothA}{
\uset{\lambda}{\gb}_{\Lambda\Gamma} = \uset{\lambda}{\Jt}_{\Lambda}^\mu \uset{\lambda}{\Jt}_{\Gamma}^\nu
S_\lambda \gt_{\mu\nu} \AND
\uset{\lambda}{\zetab}_{\Lambda} = \uset{\lambda}{\Jt}_{\Lambda}^\mu S_\lambda \zetat_{\mu},
}
respectively. From the smoothing properties of $S_\lambda$ and Sobolev's inequality,
it is clear that
\leqn{gbsmoothB}{
\uset{\lambda}{\gb}_{\Lambda\Gamma},\, \uset{\lambda}{\zetab}_{\Lambda} \in C^3([0,T)\times\Tbb^3).
}
Appealing to the familiar formula
\eqn{Gbexp}{
\Gb_{\Lambda\Gamma} = \bigl(\delta^\Sigma_\Lambda \delta^\Omega_\Gamma-\Half \gb_{\Lambda\Gamma} \gb^{\Sigma\Omega}\bigr)
\Bigl(\del{\Delta}\Gammab^\Delta_{\Sigma\Omega}-
\del{\Omega}\Gammab^\Delta_{\Sigma\Delta}+\Gammab^\Delta_{\Delta\Theta}
\Gammab^\Theta_{\Sigma\Omega}-\Gammab^\delta_{\Omega\Theta}\Gammab^\Theta_{\Delta\Sigma}\Bigr)
}
for the Einstein tensor, where
\eqn{Christbexp}{
\Gammab^\Gamma_{\Sigma\Delta} = \Half \gb^{\Gamma\Lambda}\bigl(\del{\Sigma}\gt_{\Lambda\Delta}
+\del{\Delta}\gt_{\Sigma\Lambda} -  \del{\Lambda} \gt_{\Sigma\Delta}\bigr),
}
we obtain from Propositions \ref{mollprop}, \ref{fpropA} and \ref{fpropB} the estimates
\lalin{GbestA}{
&\norm{\sqrt{\gb}\Gb_{\Lambda\Gamma}}_{\Xc^{s-2}_T(\Tbb^n)}+\Bigl\|
\uset{\;\;\;\lambda}{\sqrt{|\gb|}}\uset{\lambda}{\Gb}_{\Lambda\Gamma}\Bigr\|_{\Xc^{s-2}_T(\Tbb^n)}
\leq C\bigl(\norm{\gt}_{X^{s+1}_T(\Tbb^n)},\norm{\phiv}_{Y^{s+1}_T(\Omega)}\bigr) \label{GbestA.1}
\intertext{and}
&\Bigl\|\sqrt{\gb}\Gb_{\Lambda\Gamma}-\uset{\;\;\;\lambda}{\sqrt{|\gb|}}\uset{\lambda}{\Gb}_{\Lambda\Gamma}\Bigr\|_{\Xc^{s-2}_T(\Tbb^n)}
\leq C\bigl(\norm{\gt}_{X^{s+1}_T(\Tbb^n)},\norm{\phiv}_{Y^{s+1}_T(\Omega)}\bigr) \notag \\
&\hspace{7.0cm} \times
\bigl(\norm{\gt-S_\lambda \gt}_{X^{s+1}_T(\Tbb^n)}+\norm{\phiv-S_\lambda\phiv}_{Y^{s+1}_T(\Omega)} \bigr). \label{GbestA.2}
}
Fixing a smooth test vector field $\Yb = (\Yb^\Gamma) \in C^{\infty}_0\bigl([0,T],C^\infty(\Rbb^n,\Rbb^{n+1})\bigr)$, we
find, using the familiar formula
$\nablab^\Lambda \Yb^\Gamma$ $=$ $\gb^{\Lambda\Omega}$$\bigl(\del{\Omega}$ $\Yb^\Gamma$ $+$ $\Gammab^\Gamma_{\Omega\Sigma}$$\Yb^\Sigma\bigr)$ for the covariant derivative, and  Propositions \ref{mollprop}, \ref{fpropA} and \ref{fpropB}, that
\lalin{DYbestA}{
&\norm{\nablab^\Lambda\Yb^{\Gamma}}_{\Xc^{s-1}_T(\Tbb^n)}+\Bigl\|\uset{\lambda}{\nabla}^{\Lambda}\Yb^\Gamma
\Bigr\|_{\Xc^{s-2}_T(\Tbb^n)}
\leq C\bigl(\norm{\gt}_{X^{s+1}_T(\Tbb^n)},\norm{\phiv}_{Y^{s+1}_T(\Omega)}\bigr) \label{DYbestA.1}
\intertext{and}
&\Bigl\|\nabla^{\Lambda}\Yb^\Gamma-\uset{\lambda}{\nabla}{}_{\Lambda} \Yb^\Gamma\Bigr\|_{\Xc^{s-1}_T(\Tbb^n)}
\leq C\bigl(\norm{\gt}_{X^{s+1}_T(\Tbb^n)},\norm{\phiv}_{Y^{s+1}_T(\Omega)}\bigr) \notag \\
& \hspace{7.0cm} \times\bigl(\norm{\gt-S_\lambda \gt}_{X^{s+1}_T(\Tbb^n)}+\norm{\phiv-S_\lambda\phiv}_{Y^{s+1}_T(\Omega)} \bigr).\label{DYbestA.2}
}

Due to \eqref{gbsmoothB}, we know that classical second contracted Bianchi identity
\leqn{BianchiA}{
\uset{\lambda}{\nablab}^{\Lambda}\uset{\lambda}{\Gb_{\Lambda\Gamma}}=0
}
holds on $[0,T)\times \Tbb^n$. Since the test vector field $\Yb^{\mu}$ vanishes near $X^0=0$ and $X^0=T$,
a straightforward integration by parts arguments using Stokes' Theorem and \eqref{BianchiA} shows that
\eqn{BianchiB}{
\int_{[0,T)\times\Tbb^n} \uset{\lambda}{\nablab}^{\Lambda}\Yb^\Gamma \uset{\lambda}{\Gb}_{\Lambda\Gamma} \uset{\;\;\;\lambda}{\sqrt{|\gb|}}\, d^4 X = 0 \quad 0<\lambda \leq \lambda_0,
}
which we can use to write
\alin{BianchiC}{
\int_{[0,T)\times\Tbb^n}& \nablab^{\Lambda}\Yb^\Gamma \Gb_{\Lambda\Gamma}\sqrt{|\gb|}\, d^4 X =
\int_{[0,T)\times\Tbb^n} \nablab^{\Lambda}\Yb^\Gamma \Gb_{\Lambda\Gamma}\sqrt{|\gb|}-\uset{\lambda}{\nablab}^{\Lambda}\Yb^\Gamma \uset{\lambda}{\Gb}_{\Lambda\Gamma} \uset{\;\;\;\lambda}{\sqrt{|\gb|}}\, d^4 X \\
&=\int_{[0,T)\times\Tbb^n} \Bigl(\nablab^{\Lambda}\Yb^\Gamma-\uset{\lambda}{\nablab}^{\Lambda}\Yb^\Gamma\Bigr) \Gb_{\Lambda\Gamma}\sqrt{|\gb|}\, d^4 X  +
\int_{[0,T)\times\Tbb^n}  \uset{\lambda}{\nablab}^{\Lambda}\Yb^\Gamma\Bigl(\Gb_{\Lambda\Gamma}\sqrt{|\gb|}- \uset{\lambda}{\Gb}_{\Lambda\Gamma} \uset{\;\;\;\lambda}{\sqrt{|\gb|}}\Bigr)\, d^4 X.
}
Applying the triangle and H\"{o}lder inequalities to this expression,
the estimates \eqref{GbestA.1}-\eqref{DYbestA.2} imply that
\leqn{BianchiD}{
\left|\int_{[0,T)\times\Tbb^n} \nablab^{\Lambda}\Yb^\Gamma \Gb_{\Lambda\Gamma}\sqrt{|\gb|}\, d^4 X
\right| \leq C\bigl(\norm{\gt}_{X^{s+1}_T(\Tbb^n)},\norm{\phiv}_{Y^{s+1}_T(\Omega)}\bigr)
\bigl(\norm{\gt-S_\lambda \gt}_{X^{s+1}_T(\Tbb^n)}+\norm{\phiv-S_\lambda\phiv}_{Y^{s+1}_T(\Omega)} \bigr).
}
But, by Proposition \ref{mollprop}, we have that
\eqn{BianchiE}{
\lim_{\lambda\searrow 0} \norm{\gt-S_\lambda \gt}_{X^{s+1}_T(\Tbb^n)} = 0 \AND
\lim_{\lambda\searrow 0} \norm{\phiv-S_\lambda\phiv}_{Y^{s+1}_T(\Omega)}=0,
}
and so, we conclude from \eqref{BianchiD} that
\leqn{BianchiF}{
\int_{[0,T)\times\Tbb^n} \nablab^{\Lambda}\Yb^\Gamma \Gb_{\Lambda\Gamma}\sqrt{|\gb|}\, d^4 X =0.
}
Similar arguments starting form the classical contracted commutator identity
\eqn{commA}{
\uset{\lambda}{\nablab}^{\Lambda} \uset{\lambda}{\nablab}_{\Gamma}\uset{\lambda}{\zetab}_{\Lambda}
-\uset{\lambda}{\nablab}_{\Gamma} \uset{\lambda}{\nablab}^{\Lambda} \uset{\lambda}{\zetab}_{\Lambda}
 = \uset{\lambda}{\Rb}_{\Gamma}{}^{\Lambda}\uset{\lambda}{\zetab}_{\Lambda},
}
can be used to establish the identity
\leqn{commB}{
\int_{[0,T)\times\Tbb^n}\bigl(-\nablab^{\Lambda}\Yb^\Gamma \nablab_\Gamma \zetab_\Lambda
+\nablab_\Gamma \Yb^\Gamma \nablab^\Lambda \zetab_\Lambda\bigr)  \sqrt{|\gb|}\, d^4 X
= \int_{[0,T)\times\Tbb^n} \Yb^\Gamma \Rb_{\Gamma}{}^\Lambda\zetab_\Lambda \sqrt{|\gb|}\, d^4 X.
}

We note also that
\leqn{commC}{
\gb_{\Lambda\Gamma} \in X^s_{T}(\Tbb^n),
}
\eqn{commD}{
\zetab_{\Lambda} \in \Xc^{s-1}_T(\Tbb^n)   \cap \bigcap_{\ell=0}^1 W^{\ell,\infty}\bigl([0,T),H^{1-\ell}(\Tbb^n)\bigr),
}
and
\leqn{commE}{
\Tb_{\Lambda\Gamma} \in Y^{s}_{T}(\Omega),
}
by \eqref{gbdef}, \eqref{stressbzetab}, and Propositions \ref{fpropA} and \ref{fpropB}. It is also not difficult
 to see that \eqref{DYbestA.1}, \eqref{commC},
and \eqref{commE} guarantee
that $\Tb_{\Lambda\Gamma}$ has a well-defined trace on the boundary $(0,T)\times \del{}\Omega$, and
that $\gb_{\Lambda\Gamma}$, $\nablab^{\Lambda}\Yb^{\Gamma}$,
and $\Tb_{\Lambda\Gamma}$
are regular enough to
justify the following integration by parts argument:
\lalin{divTb}{
\int_{[0,T)\times \Tbb^n} \nablab^{\Lambda}&\Yb^{\Gamma}\chi_\Omega \Tb_{\Lambda\Gamma} \sqrt{|\gb|}\, d^4 X
= \int_{[0,T)\times \Omega} \nablab^{\Lambda}\Yb^{\Gamma}\Tb_{\Lambda\Gamma} \sqrt{|\gb|}\, d^4 X \notag \\
&= \int_{[0,T)\times \del{}\Omega} \nb^\Lambda \Yb^{\Gamma}\Tb_{\Lambda\Gamma}\, d\bar{\mu}
- \int_{[0,T)\times \Omega} \Yb^{\Gamma} \nablab^\Lambda \Tb_{\Lambda\Gamma} \sqrt{|\gb|}\, d^4 X
= 0, \label{divTb.1}
}
where in obtaining the last equality we have used \eqref{EEbarA.2} and \eqref{EEbarB}.

Contracting \eqref{EEbarA.1} with $\nablab^\Lambda \Yb^{\Gamma}$ and integrating over $[0,T)\times \Tbb^n$,
we observe, with help of  \eqref{BianchiF}, \eqref{commB} and \eqref{divTb.1}, that
$\zetab_{\Lambda}$ satisfies
\eqn{weakgaugeA}{
\int_{[0,T)\times \Tbb^n} \bigl(-\nablab^{\Lambda}\Yb^{\Gamma} \nablab_{\Lambda}\zetab_{\Gamma}
+ \Yb^{\Gamma}\Rb_\Gamma{}^\Lambda \zetab_{\Lambda} \bigr) \sqrt{|\gb|}\, d^4 X.
}
Since $\Yb^{\Gamma}$ was chosen arbitrarily, it follows that $\zetab_\Lambda $$\in$ $\bigcap_{\ell=0}^1$ $W^{\ell,\infty}\bigl([0,T),H^{1-\ell}(\Tbb^n)\bigr)$ is a weak
solution of the IVP:
\alin{weakgaugeB}{
\nablab^{\Lambda} \nablab_{\Lambda}\zetab_{\Gamma}
+ \Rb_\Gamma{}^\Lambda \zetab_{\Lambda} & = 0 \hspace{1.0cm} \text{in $[0,T)\times \Tbb^n$,}\\
(\zetab_\Lambda,\del{0}\zetab_\Lambda) &= (0,0) \hspace{0.4cm} \text{in $\{0\}\times \Tbb^n$.}
}
By the uniqueness of weak solutions, we conclude that $\zetab_\Lambda = 0$ in $[0,T)$$\times$$\Tbb^n$, which
is clearly equivalent to $\zetat_\mu = 0$ in $[0,T)$$\times$$\Tbb^n$. This complete the proof of Theorem \ref{fulllocthm}.

\newpage

\noindent \emph{\textbf{Acknowledgments}}

\smallskip

\noindent This work was partially supported by the ARC grant FT1210045.
Part of this work was completed during a visit by the authors T.A.O. and B.G.S.
to the Albert Einstein Institute. We are grateful to the Institute for
its support and hospitality during these visits.
\appendix

\sect{calculus}{Calculus inequalities}

In this appendix, we state, for the convenience of the reader, a number of calculus inequalities that will
be used throughout this article. In the following, $\Omega$ will always denote an open subset of $\Tbb^n$ with smooth boundary.
Proofs of the inequalities involving the standard Sobolev spaces $W^{s,p}(\Omega)$ are well known and may be found, for example, in the books \cite{AdamsFournier:2003}, \cite{Friedman:1976} and \cite{TaylorIII:1996}.
The proofs of the inequalities involving the $\Hc^{0,s}(\Tbb^n)$ space can either be found in Appendix A of \cite{AnderssonOliynyk:2014}, or else are
easily derived from the inequalities found there.

\begin{thm}{\emph{[Sobolev's inequality]}} \label{Sobolev} Suppose $s\in \Zbb_{\geq 1}$,
$1\leq p < \infty$, and $sp>n$. Then
\eqn{Sobolev1}{
\norm{u}_{L^{\infty}(\Omega)} \lesssim \norm{u}_{W^{s,p}(\Omega)}
}
for all $u\in W^{s,p}(\Omega)$.
\end{thm}

\begin{thm}{\emph{[Integral Rellich-Kondrachov Theorem]}} \label{iRellich}
Suppose $1\leq p \leq n$ and let
\eqn{iRellich1}{
p^*= \begin{cases} \frac{np}{n-p} & \text{if $p<n$}\\
\infty & \text{if $p=n$}
\end{cases}.
}
Then $W^{1,p}(\Omega)\subset L^q(\Omega)$ for
$1\leq q < p^*$, and the embedding is compact.
\end{thm}

\begin{thm}{\emph{[Fractional Rellich-Kondrachov Compactness Theorem]}} \label{fRellich}
Suppose $s\in (0,1)$, $1<p<\infty$, and let
\eqn{fRellich1}{
p_s= \begin{cases} \frac{np}{n-sp} & \text{if $sp<n$}\\
\infty & \text{if $p=sn$}
\end{cases}.
}
Then
$W^{s,p}(\Omega)\subset L^q(\Omega)$ for $1\leq q < p_s$, and the embedding is compact.
\end{thm}

\begin{thm}{\emph{[Multiplication estimates]}} \label{calcpropB} $\;$

\begin{enumerate}[(i)]
\item If $1\leq p <\infty$, $s_1,s_2,s_{3}\in \Zbb$, $s_1,s_2 \geq s_{3}\geq 0$, and $s_1+s_2 -n/p > s_{3}$, then
\eqn{calcpropB.1}{
\norm{u_1 u_2}_{W^{p,s_3}(\Omega)} \lesssim \norm{u_1}_{W^{p,s_1}(\Omega)} \norm{u_2}_{W^{p,s_2}(\Omega)}
}
for all $u_i\in W^{p,s_i}(\Omega)$, $i=1,2$.
\item If  $s_1,s_2,s_{3}\in \Zbb$, $s_1,s_2 \geq s_{3}\geq 0$, and $s_1+s_2 -n/2 > s_{3}$, then
\eqn{calcproB.2}{
\norm{u_1 u_2}_{\Hc^{0,s_3}(\Tbb^n)} \lesssim \norm{u_1}_{\Hc^{0,s_1}(\Tbb^n)} \norm{u_2}_{\Hc^{0,s_2}(\Tbb^n)}
}
for all $u_i \in \Hc^{0,s_i}(\Tbb^n)$.
\end{enumerate}
\end{thm}


\begin{thm}{\emph{[Moser's estimates]}}  \label{Moser}
Suppose $s\in \Zbb_{\geq 1}$, $1\leq p \leq \infty$, $|\alpha|\leq s$, $f\in C^s(\Rbb)$, $f(0) = 0$,
 and $g\in C^{s+1}(\Rbb)$. Then
\eqn{Moser.1}{
\norm{D^\alpha f(u)}_{L^{p}(\Omega)} \leq C\bigl(\norm{f}_{C^s}\bigr)(1+\norm{u}^{s-1}_{L^\infty(\Omega)})\norm{u}_{W^{s,p}(\Omega)}
}
and
\eqn{Moser.2}{
\norm{D^\alpha ( g(u)- g(v))}_{L^{p}(\Rbb^n)} \leq C\bigl(\norm{g}_{C^{s+1}}\bigr)(1+\norm{u}^{s-1}_{L^\infty(\Rbb^n)}+ \norm{v}^{s-1}_{L^\infty(\Omega)} )\norm{u-v}_{W^{s,p}(\Omega)}
}
for all $u,v \in L^\infty(\Omega)\cap W^{s,p}(\Omega)$.
\end{thm}

\begin{prop} \label{fpropA}
Suppose $s_1,s_2,s_{3}\in \Zbb$, $s_1,s_2 \geq s_{3}\geq 0$, $s_1+s_2 -n/2 > s_{3}$, and $0\leq \ell \leq s_3$. Then
\eqn{fpropA1}{
\norm{\del{t}^\ell (u_1 u_2)}_{\Hc^{0,s_3-\ell}(\Tbb^n)} \lesssim \norm{u_1}_{\Ec^{s_1}(\Tbb^n)} \norm{u_2}_{\Ec^{s_2}(\Tbb^n)}
}
and
\eqn{fpropA2}{
\norm{\del{t}^\ell (v_1 v_2)}_{H^{s_3-\ell}(\Omega)} \lesssim \norm{v_1}_{E^{s_1}(\Omega)} \norm{v_2}_{E^{s_2}(\Omega)}
}
for all $u_i \in \Ec^{s_i}(\Tbb^n)$ and $v_i \in E^{s_i}(\Omega)$, $i=1,2$.
\end{prop}

\begin{prop} \label{fpropB} Suppose $s\in \Zbb_{> n/2}$, $f\in C^s(\Rbb)$, $f(0)=0$, and
$g\in C^{s+1}(\Rbb)$. Then
\gath{propB1}{
\norm{\del{t}^\ell f(u)}_{\Hc^{0,s-\ell}(\Tbb^n)} \leq C\bigl(\norm{u}_{\Ec^s(\Tbb^n)}\bigr)\norm{u}_{\Ec^s(\Tbb^n)},\\
\norm{\del{t}^\ell f(u)-\del{t}^\ell f(v)}_{\Hc^{0,s-\ell}(\Tbb^n)}
\leq C\bigl(\norm{u}_{\Ec^s(\Tbb^n)},\norm{v}_{\Ec^s(\Tbb^n)}\bigr)\norm{u-v}_{\Ec^s(\Tbb^n)}
\intertext{and}
\norm{\del{t}^\ell f(v)}_{H^{s-\ell}(\Omega)} \leq C\bigl(\norm{u}_{E^s(\Omega)}\bigr)\norm{u}_{E^s(\Omega)}
}
for all $u,v \in \Xc_T^s(\Tbb^n)$, $w \in Y_T^s(\Omega)$ and  $0\leq \ell \leq s$.
\end{prop}

\bibliographystyle{amsplain}
\bibliography{refs}

\end{document}